\documentclass[preprint,12pt]{elsarticle}

\usepackage{graphicx}%
\usepackage{multirow}%
\usepackage{amsmath,amssymb,amsfonts}%
\usepackage{amsthm}%
\usepackage{mathrsfs}%
\usepackage[title]{appendix}%
\usepackage{xcolor}%
\usepackage{textcomp}%
\usepackage{manyfoot}%
\usepackage{booktabs}%
\usepackage{algorithm}%
\usepackage{algorithmicx}%
\usepackage{algpseudocode}%
\usepackage{listings}%
\usepackage{comment}
\DeclareFontFamily{U}{mathx}{}
\DeclareFontShape{U}{mathx}{m}{n}{<-> mathx10}{}
\DeclareSymbolFont{mathx}{U}{mathx}{m}{n}
\DeclareMathAccent{\widehat}{0}{mathx}{"70}
\DeclareMathAccent{\widecheck}{0}{mathx}{"71}
\biboptions{sort&compress}

 \newtheorem{lemma}{Lemma}
  \newtheorem{corollary}{Corollary}
\newtheorem{theorem}{Theorem}%

\newtheorem{proposition}[theorem]{Proposition}%

\newtheorem{example}{Example}%
\newtheorem{remark}{Remark}%

\newtheorem{definition}{Definition}%

\usepackage{physics}
\usepackage{xcolor,soul}

 \newcommand{\x}{ {\bf x}}
  
\newcommand{\z}{ {\bf z}}
\DeclareFontFamily{U}{dice3d}{}
\DeclareFontShape{U}{dice3d}{m}{n}{<-> s*[9] dice3d}{}

\usepackage{phaistos}
\usepackage{tikz}
\newcommand*\coin[1]{\tikz[baseline=(char.base),scale=0.7, transform shape]{
            \node[shape=circle,draw,minimum size=2mm] (char) {#1};}}

\journal{}

\makeatletter
\def\ps@pprintTitle{%
  \let\@oddhead\@empty
  \let\@evenhead\@empty
  \let\@oddfoot\@empty
  \let\@evenfoot\@oddfoot
}
\makeatother
\begin{document}

\begin{frontmatter}

\title{Connecting classical finite exchangeability to quantum theory and indistinguishability}

\author[1,2]{Alessio Benavoli} %
\author[3,4]{Alessandro Facchini} %
\author[3]{Marco Zaffalon} %
\affiliation[1]{organization={School of Computer Science and Statistics, Trinity College Dublin},%
            city={Dublin},
            country={Ireland}}

\affiliation[2]{organization={Trinity Quantum Alliance, Trinity College Dublin},%
            addressline={Unit 16 Trinity Technology and Enterprise Centre}, 
            city={Dublin},
            country={Ireland}}

\affiliation[3]{organization={SUPSI, DTI, Istituto Dalle Molle di Studi sull'Intelligenza Artificiale (USI-SUPSI)},%
            city={Lugano},
            country={Switzerland}}
            
\affiliation[4]{organization={Management in Networked and Digital Societies (MINDS) Department, Kozminski University},
      \city{Warsaw},\country{Poland}}

\begin{abstract}
Exchangeability is a fundamental concept in probability theory and statistics. It allows to model situations where the order of observations does not matter. The classical de Finetti's theorem provides a representation of infinitely exchangeable sequences of random variables as mixtures of independent and identically distributed variables. The quantum de Finetti theorem extends this result to symmetric quantum states on tensor product Hilbert spaces. It is well known that both theorems do not hold for finitely exchangeable sequences. The aim of this work is to investigate two lesser-known representation theorems, which were developed in classical probability theory to extend de Finetti's theorem
to finitely exchangeable sequences by using quasi-probabilities and quasi-expectations. With the aid of these theorems, we  illustrate how a de Finetti-like representation theorem for finitely exchangeable sequences can be formulated through a mathematical representation which is formally equivalent to quantum theory
(with boson-symmetric density matrices). We then show a promising application of this connection to the challenge of defining entanglement for indistinguishable bosons.
\end{abstract}


\end{frontmatter}

\section{Introduction}
\label{sec:intro}

Let $(t_1 , t_2,\dots,t_r)$ be a sequence of $r$ random variables, each variable taking values in a finite possibility space $\Omega$, and let $P(t_1 , t_2,\dots,t_r)$ be their joint probability distribution. When $P(t_1 , t_2,\dots,t_r)$  is symmetric to permutations of the labels of the variables $t_i$, the sequence is called finitely exchangeable with respect to $P$.
The classical de Finetti theorem states that, whenever the considered sequence is finitely exchangeable with respect to $P$ and for every $s > r$ there is a finitely exchangeable sequence $(t_1 , t_2,\dots,t_s)$ with respect to a joint probability $P'$  such that $P$ coincides with the marginal of $P'$ on $(t_1,t_2,\dots,t_{r})$,
then  $P(t_1 , t_2,\dots,t_r)$ can be represented as:
\begin{equation}
\label{eq:cdF}
\begin{aligned}
 P(t_1 , t_2,\dots,t_r)&= E[P(t_1)P( t_2)\cdots P(t_r)]\\
 &:=\int P(t_1)P( t_2)\cdots P(t_r) dq(P),
\end{aligned}
\end{equation}
 with $q$ a probability measure on the marginal probability distribution $P(t)$ (of one variable).

One well known issue with this theorem is that it necessarily requires an infinite  sequence extending $(t_1 , t_2,\dots,t_r)$ for which exchangeability holds for each segment longer than $r$. Without such requirement, and thus by sticking to a finitely exchangeable sequence only, the representation of the operator $E$ in Eq.~ \eqref{eq:cdF} as a probabilistic mixture of identical product distributions fails in general. However, if the finitely exchangeable sequence of $r$ variables is part of a finite but large enough finitely excheangable sequence of $s >> r$ variables, the joint distribution of the initial $r$ random
variables can be well-approximately represented in a de Finetti-like manner. Much of the work on finite exchangeability in classical probability has actually focused on deriving analytical bounds for the error of this approximation \cite{diaconis1977finite,diaconis1980finite}. 

A quantum analogue of de Finetti's theorem has been derived in \cite{hudson1976locally,caves2002unknown,renner2007symmetry,christandl2007one,konig2005finetti,koenig2009most}. Widely used in quantum information theory, it states the following. Let $\rho^{(r)}$ be  a density matrix on the $r$-tensor product Hilbert space $\mathcal{H}^{\otimes r}$ ({ if $\dim(\mathcal{H})=d$, then $\rho^{(r)}$ is a $d^r \times d^r$ complex Hermitian  positive semi-definite matrix with trace one}). Whenever $\rho^{(r)}$ is symmetric, that is it is such that $\Pi_{sym}\rho^{(r)} \Pi_{sym}^{\dagger} = \rho^{(r)}$ with the matrix $\Pi_{sym}$ being the so-called `symmetriser',  and there is a symmetric density matrix $\rho^{(s)}$ on  $\mathcal{H}^{\otimes s}$ such that $\rho^{(r)} = Tr_{s-r}(\rho^{(s)})$ for
every $s > r$, where $Tr_{s-r}$ denotes the partial-trace ({ which performs a marginalisation operation over the subset of the 
$s-r$ subsystems}), then it holds that
\begin{equation}
\label{eq:qdF}
  \rho^r = \int (\sigma^{\otimes s }) dq(\sigma) ,
\end{equation}
with $\sigma$ a density matrix in $\mathcal{H}$, { $\sigma^{\otimes s }$ denotes the tensor product of $\sigma$ with itself $s$-times},  and $q$ a probability measure on the set of all density matrices in $\mathcal{H}$.
The research on analytical bounds has also been extended to the quantum setting \cite{renner2007symmetry,christandl2007one} with applications to quantum cryptography  \cite{renner2008security,su2020improved}, emergence of classical reality \cite{brandao2015generic} and, more recently,  the black hole information puzzle  \cite{renner2021black}.

In this paper we investigate two other representation theorems that extend the classical de Finetti's theorem so as to cover also the case of finitely exchangeable sequences. Both modify the definition of operator $E$ in Eq.~\eqref{eq:cdF}. The first solution keeps the same representation as in the original theorem, except that $q$ is a normalised signed-measure (a quasi-probability) \cite{dellacherie1982probabilities,jaynes1982some,kerns2006definetti} (which also holds true in the quantum de-Finetti case \cite{christandl2007one,Ammari2008,Lewin2014}). The second solution instead employs a \emph{quasi}-expectation operator applied to certain polynomial functions \cite{jaynes1982some,kerns2006definetti,harper2007probability,de2012exchangeability,zaffalon2019b}. 
Our objective is to illustrate, with the aid of the above results, how de Finetti-like representation theorems for finitely exchangeable systems  can be formulated using a mathematical representation which is formally equivalent to quantum theory (QT for short). This objective is achieved in two steps. 
First, we will use polynomials of complex conjugate variables to re-derive the extension of the classical de Finetti's theorem based on \textit{quasi}-expectation operators but expressed with the same mathematical objects as in QT.
Second, we will  demonstrate that such \textit{quasi}-expectation operators are equivalent to boson-symmetric density matrices. In doing so, we will be able to simulate finite exchangeable sequences using quantum experiments with bosons.
In summary, we will show that a de Finetti-like representation of finite exchangeability leads to a representation theorem which encompasses both classical and quantum theories as special cases. 
We will then discuss how such representation is related to previous results 
\cite{benavoli2021indi,BENAVOLI2022389,doherty2004complete} and to
other formulations of QT involving  quasi-probability densities and expectations. %
Finally, we will use this representation to provide fresh insights into the challenge of defining entanglement for indistinguishable bosons. There are at
least five main different definitions of entanglement for indistinguishable bosons in literature.  Recently, \cite{benatti2020entanglement} attempted to systematise these different definitions by proposing a set of properties that should be used to characterise entanglement for indistinguishable particles. Through our
connection between QT and classical finitely exchangeable probabilities, we will show that the definitions of two properties given in \cite{benatti2020entanglement} lead to clearly interpretable contradictions when applied to classical finitely exchangeable probabilities. For instance, one of these properties, \textit{local operators}, would imply that the uniform distribution is finitely but not infinitely exchangeable, which is clearly false. We then use this analysis to reformulate this property in a way that avoids this issue.\\
To provide a clearer and more practical understanding of the different representation theorems, in what follows we use a dice rolling experiment as a case study. These examples  will later help to clarify the application of exchangeability to the problem of defining entanglement for indistinguishable bosons. 

\section{The classical de Finetti's representation theorem}
Consider a dice whose possibility space is given by its six faces:
$$\Omega:=\Big\{\begin{gathered}\text{\scriptsize \usefont{U}{dice3d}{m}{n} \text{ 3a}}, \text{\scriptsize \usefont{U}{dice3d}{m}{n} \text{3b}}, \text{\scriptsize \usefont{U}{dice3d}{m}{n} \text{2b}}, \text{\scriptsize \usefont{U}{dice3d}{m}{n} \text{2c}},\text{\scriptsize \usefont{U}{dice3d}{m}{n}  \text{3c}}, \text{\scriptsize \usefont{U}{dice3d}{m}{n} \text{2d} }\end{gathered}\Big\}.
 $$ 
 Let us assume that the upward-facing side of the dice is the one that is facing the reader. By $t_1,t_2,\dots,t_r$ we denote the results of $r$-rolls of the dice. For example, for $r=3$, we can have the three outcomes:
   \begin{center}
 {\usefont{U}{dice3d}{m}{n}$\stackrel{\text{\scriptsize 3b}}{t_1}$ $\stackrel{\text{\scriptsize 2b}}{t_2}$
 $\stackrel{\text{\scriptsize 2b}}{t_3}$}
\end{center}
that is $t_1=2$, $t_2=3$, and $t_3=3$.\\
\begin{definition}
\label{def:finexcg}
Let $t_1,t_2,\dots,t_r$ a sequence of variables and $P$ a joint probability. The sequence is said to be \textit{finitely exchangeable} with respect to $P$, if 
$$
P(t_1,t_2,\dots,t_r)=P(t_{\pi(1)},t_{\pi(2)},\dots,t_{\pi(r)}),
$$
for any permutation $\pi$ of the index sequence $(1,2,\dots,r)$.\\
\end{definition}
Considering the previous example. To say that the three rolls are finitely exchangeable means that 
$$
 P\left(\begin{gathered}\text{\tiny \usefont{U}{dice3d}{m}{n}$\stackrel{\text{3b}}{}$ $\stackrel{\text{2b}}{}$
 $\stackrel{\text{2b}}{}$}\end{gathered}\right)= P\left(\begin{gathered}\text{\tiny \usefont{U}{dice3d}{m}{n}$\stackrel{\text{2b}}{}$ 
 $\stackrel{\text{3b}}{}$
 $\stackrel{\text{2b}}{}$}\end{gathered}\right)=P\left(\begin{gathered}\text{\tiny \usefont{U}{dice3d}{m}{n}$\stackrel{\text{2b}}{}$ 
 $\stackrel{\text{2b}}{}$
 $\stackrel{\text{3b}}{}$}\end{gathered}\right).
$$

\begin{definition}
\label{def:excg}
The sequence is called
 \textit{infinitely exchangeable} (with respect to $P$) if it is finitely exchangeable  (with respect to $P$) and if for any  $s > r$,
there is a distribution $P'$ with respect to which $t_1,t_2,\dots,t_{s}$ is finitely exchangeable and 
such that
$$
P(t_1,\dots,t_{r})=\sum_{t_{r+1},\dots, t_{s}} P'(t_{1},\dots,t_r,t_{r+1},\dots, t_{s}).
$$
\end{definition}
Stated otherwise, 
an infinitely exchangeable sequence can be understood as a subsequence 
{\small
$$
{\dots,t_{k-1},\,t_{k}},\,\overbrace{{t_{k+1},\,t_{k+2},\,\dots,\,t_{k+n-1},\,t_{k+r},}}^{\text{ subsequence of length $r$}}\,{t_{k+r+1},\dots}  
$$}
\\ {
of an infinite sequence of random variables whose joint distribution is invariant under permutations, meaning that their order is irrelevant.
}

 Given $r$ exchangeable rolls $t_1,t_2,\dots,t_r$, their output 
 is fully characterised by the counts:
    \begin{center}
 {\usefont{U}{dice3d}{m}{n}$\stackrel{\text{\scriptsize 3a}}{n_1}$ $\stackrel{\text{\scriptsize 3b}}{n_2}$
 $\stackrel{\text{\scriptsize 2b}}{n_3}$
  $\stackrel{\text{\scriptsize 2c}}{n_4}$
   $\stackrel{\text{\scriptsize 3c}}{n_5}$
    $\stackrel{\text{\scriptsize 2d}}{n_6}$
 }
\end{center}
where $n_i$ denotes the number of times the dice landed on i-th face in the $r$-rolls. 
Based on this fact, de Finetti  proved his famous \textit{representation theorem}.\footnote{De Finetti actually proved his theorem for the binary case, a coin. His result can easily be extended to the dice.}\\
\begin{proposition}[\cite{finetti1937}]
\label{prop:definetti}
Let  $\Omega$ be the possibility spaces given by the six faces of a dice, and let $(t_{1}, \dots, t_{r})$ be a sequence of random variables that is infinitely exchangeable with probability  measure $P$. Then there exists a distribution function $q$ such that
\begin{equation}
    \label{eq:reprTh0}
P(t_{1},\dots,t_{r}) = \int_\Theta \theta_1^{n_1}\theta_2^{n_2}\cdots \theta_5^{n_5}\theta_6^{n_6} dq({\boldsymbol \theta}), 
\end{equation}
where ${\boldsymbol \theta}=[\theta_1,\theta_2,\dots,\theta_6]^\top$ are the probabilities of the corresponding faces and whose values belong to the possibility space: 
$$
\Theta=\left\{\boldsymbol{\theta} \in \mathbb{R}^6: \theta_i\geq0,~\sum_{i=1}^6 \theta_i=1\right\},
$$ 
and $n_i=\sum_{j=1}^r \delta_{i}(t_j)$, where $\delta_{a}(b)$ is the Kronecker delta\footnote{The Kronecker delta is defined as follows $\delta_a(b)=1$ if $a=b$ and zero otherwise.}, that is $n_i$ is the number of occurrences of the i-th face in the sequence $t_1,\dots,t_r$ (this means that $\sum_{i=1}^6 n_i=r$).\\
\end{proposition}
In the context of Bayesian statistics,  the theorem states that any infinitely exchangeable sequence can be seen as generated by first selecting a probability distribution $\boldsymbol{\theta}$ from some prior $q$ (probability distribution on the space of probability distributions) and then letting the $t_i$ to be i.i.d.\ with common distribution $\boldsymbol{\theta}$.
 For a selected ${\boldsymbol \theta}$, this for instance  means that $P(t_1=2,t_2=3,t_3=3)=\theta_2\theta_3^2$ (the product comes from the independence assumption).  

As mentioned in Section \ref{sec:intro}, one known issue with the above  theorem is that the assumption of an infinite exchangeable sequence is necessary for the specific representation to hold \cite{deFinetti1930,diaconis1980}.
As a counter-example, consider:
\begin{equation}
     \label{eq:psym}
\begin{aligned}
P(t_1=i,t_2=i)&=0,\\~~   P(t_1=i,t_2=j)&= P(t_1=j,t_2=i)=\frac{1}{30}\\
\end{aligned}
 \end{equation}
 for all $i \neq j =1,2,\dots,6$. The variables $t_1,t_2$ are { finitely} exchangeable. Assume that a representation like \eqref{eq:reprTh0} holds. Then   
 $$
 0=P(t_1=i,t_2=i)=\int_\Theta \theta_i^{2} dq({\boldsymbol \theta}) 
 $$ for all $i$, would imply that $q$ puts mass $1$ at all the $\theta_i=0$, which is impossible. This means that, although the $P$ in \eqref{eq:psym} is a valid probability, we cannot express it as \eqref{eq:reprTh0}. In other words, exchangeability is not equivalent to independence.

Examples of situations that are truly finitely exchangeable were already discussed by de Finetti in \cite{deFinetti1930}. They can be illustrated by the act of selecting items from an urn without replacement.\\

\begin{example}
\label{ex:urn}
Suppose we have an urn containing one black (B) ball and one white (W) ball. Draw out balls, one at a
 time and without replacement, and note the colour. 
 Then we have that
 
 $$
\begin{aligned}
P(t_1=B,t_2=B)&=0,\\
P(t_1=B,t_2=W)&=P(t_1=W,t_2=B)=0.5,\\
P(t_1=W,t_2=W)&=0,
\end{aligned}
 $$
which is finitely but not infinitely exchangeable. 

In QT, we can rule out this type of urn-model dependence by choosing the direction of the measurement randomly and independently for each particle. This is done in Bell's experiments. This comment is relevant for the relationship between exchangeability and entanglement, which we discuss in Section \ref{sec:Ent}.\\
\end{example}

Kendall \cite{kendall1966finite} proved that every finite system of exchangeable events is indeed equivalent to a random sampling scheme without replacement, where the number of items in the sampling has an arbitrary distribution.

While it is evident that de Finetti's theorem can be invalid for finite exchangeable sequences, for a large enough integer $s>r$, the distribution of the initial $r$ random variables of an exchangeable vector of length  $s$ can be approximately represented as in Eq. \eqref{eq:reprTh0}. As previously mentioned, much of the work on finite exchangeability has focused on deriving analytical bounds for the error of this approximation \cite{diaconis1977finite,diaconis1980finite}. 
To clarify the meaning of the term `bound'  in this context, we use the dice-rolling example to discuss the general problem of bounding a linear function of event probabilities, where the underlying probability measure $P_r$ is a probability distribution  of a finitely exchangeable sequence  $t_1,\dots,t_r$. In the rest of the manuscript, the subscript $r$ in $P_r$ will be used  to indicate the length of the sequence $t_1,\dots,t_r$.\\

\begin{example}
\label{eq:diffP}
Consider $r=2$ (two rolls) and assume, for example, we aim to find a lower bound for the following quantity
{\small
\begin{equation}
\label{eq:boundP2}
P_2\left(\begin{gathered}\text{\tiny \usefont{U}{dice3d}{m}{n}$\stackrel{\text{3a}}{}$ $\stackrel{\text{3a}}{}$}\end{gathered}\right)- P_2\left(\begin{gathered}\text{\tiny \usefont{U}{dice3d}{m}{n}$\stackrel{\text{3a}}{}$ 
 $\stackrel{\text{3b}}{}$
}\end{gathered}\right)+P_2\left(\begin{gathered}\text{\tiny \usefont{U}{dice3d}{m}{n}$\stackrel{\text{3b}}{}$ 
 $\stackrel{\text{3b}}{}$}\end{gathered}\right),
\end{equation}
}
where $P_2$ is a probability distribution  of the finitely exchangeable sequence  $t_1,t_2$ (two rolls). In other words, among all  probability measures $P_2$ we aim to find the one which minimises \eqref{eq:boundP2}. We will focus only on the lower bound (minimisation) as an example.

Therefore, our goal is to compute:
{\small
$$
v_2=\min_{P_2} \Bigg( P_2\left(\begin{gathered}\text{\tiny \usefont{U}{dice3d}{m}{n}$\stackrel{\text{3a}}{}$ $\stackrel{\text{3a}}{}$}\end{gathered}\right)- P_2\left(\begin{gathered}\text{\tiny \usefont{U}{dice3d}{m}{n}$\stackrel{\text{3a}}{}$ 
 $\stackrel{\text{3b}}{}$
}\end{gathered}\right)+P_2\left(\begin{gathered}\text{\tiny \usefont{U}{dice3d}{m}{n}$\stackrel{\text{3b}}{}$ 
 $\stackrel{\text{3b}}{}$}\end{gathered}\right)\Bigg).
$$
}
It is not difficult to see that the minimum is $v_2=-0.5$ and the  finitely exchangeable $P_2$ which achieves it is given by setting $P_2\left(\begin{gathered}\text{\tiny \usefont{U}{dice3d}{m}{n}$\stackrel{\text{3a}}{}$ 
 $\stackrel{\text{3b}}{}$
}\end{gathered}\right)=P_2\left(\begin{gathered}\text{\tiny \usefont{U}{dice3d}{m}{n}$\stackrel{\text{3b}}{}$ 
 $\stackrel{\text{3a}}{}$
}\end{gathered}\right)=0.5$ and zero otherwise.

Now assume that these two rolls are part of a finitely exchangeable sequence of three rolls, and let $P_{3}$ be the relative finitely exchangeable probability. In this case, we aim to solve  the  problem
{\small
$$
v_3=\min_{P_3} \Bigg(P_3\left(\begin{gathered}\text{\tiny \usefont{U}{dice3d}{m}{n}$\stackrel{\text{3a}}{}$ $\stackrel{\text{3a}}{}$}\end{gathered}\right)- P_3\left(\begin{gathered}\text{\tiny \usefont{U}{dice3d}{m}{n}$\stackrel{\text{3a}}{}$ 
 $\stackrel{\text{3b}}{}$
}\end{gathered}\right)+P_3\left(\begin{gathered}\text{\tiny \usefont{U}{dice3d}{m}{n}$\stackrel{\text{3b}}{}$ 
 $\stackrel{\text{3b}}{}$}\end{gathered}\right)\Bigg),
$$
}
where $P_3(t_1,t_2)$ is computed from $P_3(t_1,t_2,t_3)$ via marginalisation, that is $P_3(t_1=a,t_2=b)=\sum_{i=1}^6P_3(t_1=a,t_2=b,t_3=i)$ for every $a,b \in \{1,2,\dots,6\}$.
Note that, since $P_3$ is the probability distribution of the exchangeable sequence $t_1,t_2,t_3$ then $P_3(t_1=a,t_2=b,t_3=i)=P_3(t_1=b,t_2=a,t_3=i)=P_3(t_1=i,t_2=b,t_3=a)$ and so on for all the $3!$ permutations. In this case, computing $v_3$ is more challenging because exchangeability imposes many constraints on $P_3$, such as:
{\small
$$
\begin{aligned}
P_3\left(\begin{gathered}\text{\tiny \usefont{U}{dice3d}{m}{n}$\stackrel{\text{3a}}{}$ $\stackrel{\text{3b}}{}$ $\stackrel{\text{3a}}{}$}\end{gathered}\right)&=P_3\left(\begin{gathered}\text{\tiny \usefont{U}{dice3d}{m}{n}$\stackrel{\text{3a}}{}$ $\stackrel{\text{3a}}{}$ $\stackrel{\text{3b}}{}$}\end{gathered}\right)=P_3\left(\begin{gathered}\text{\tiny \usefont{U}{dice3d}{m}{n}$\stackrel{\text{3b}}{}$ $\stackrel{\text{3a}}{}$ $\stackrel{\text{3a}}{}$}\end{gathered}\right).
\end{aligned}
$$
}
In particular, this means that, in this case,  to compute $v_3$ we cannot set  $P_3\left(\begin{gathered}\text{\tiny \usefont{U}{dice3d}{m}{n}$\stackrel{\text{3a}}{}$ $\stackrel{\text{3a}}{}$}\end{gathered}\right)$ equal to zero. Indeed, since  its value depends on the value of $P_3\left(\begin{gathered}\text{\tiny \usefont{U}{dice3d}{m}{n}$\stackrel{\text{3a}}{}$ $\stackrel{\text{3a}}{}$ $\stackrel{\text{3b}}{}$}\end{gathered}\right)$ via marginalisation, then if $P_3\left(\begin{gathered}\text{\tiny \usefont{U}{dice3d}{m}{n}$\stackrel{\text{3a}}{}$ $\stackrel{\text{3a}}{}$}\end{gathered}\right)$  were zero, than also $P_3\left(\begin{gathered}\text{\tiny \usefont{U}{dice3d}{m}{n}$\stackrel{\text{3a}}{}$ $\stackrel{\text{3b}}{}$}\end{gathered}\right)$ would be zero. We will later show that in this case the solution of the minimisation problem is $v_3\approx -0.167$.

We can also show that when the two rolls are part of an infinitely exchangeable sequence of rolls it holds that
{\small
$$
\begin{aligned}
v_{\infty}&=\inf_{P_\infty} \Bigg(P_\infty\left(\begin{gathered}\text{\tiny \usefont{U}{dice3d}{m}{n}$\stackrel{\text{3a}}{}$ $\stackrel{\text{3a}}{}$}\end{gathered}\right)- P_\infty\left(\begin{gathered}\text{\tiny \usefont{U}{dice3d}{m}{n}$\stackrel{\text{3a}}{}$ 
 $\stackrel{\text{3b}}{}$
}\end{gathered}\right)+P_\infty\left(\begin{gathered}\text{\tiny \usefont{U}{dice3d}{m}{n}$\stackrel{\text{3b}}{}$ 
 $\stackrel{\text{3b}}{}$}\end{gathered}\right)\Bigg)\\
 &=0.
 \end{aligned}
$$
}
In summary, we can compute  lower bounds of $P_n\left(\begin{gathered}\text{\tiny \usefont{U}{dice3d}{m}{n}$\stackrel{\text{3a}}{}$ $\stackrel{\text{3a}}{}$}\end{gathered}\right)- P_n\left(\begin{gathered}\text{\tiny \usefont{U}{dice3d}{m}{n}$\stackrel{\text{3a}}{}$ 
$\stackrel{\text{3b}}{}$
}\end{gathered}\right)+P_n\left(\begin{gathered}\text{\tiny \usefont{U}{dice3d}{m}{n}$\stackrel{\text{3b}}{}$ 
$\stackrel{\text{3b}}{}$}\end{gathered}\right)$ for any $n\geq 2$  and  these bounds satisfy the inequalities
$$
-0.5=v_2 < v_3 < \dots < v_\infty=0.
$$
This means that as $n$ increases, $v_n$ approaches the value computed using the probability $P_{\infty}$ of an infinitely exchangeable sequence. 
 In Section \ref{sec:Bern}, we will discuss a general numerical algorithm to compute a lower bound of any linear function of event probabilities. 
 
 It is important to emphasize that the inequalities $v_2 < v_3 < \dots < v_\infty$ do not hold for every linear function 
of event probabilities. For example, if we consider the lower bound $w_{n}=\min_{P_n} \Big(P_n\left(\begin{gathered}\text{\tiny \usefont{U}{dice3d}{m}{n}$\stackrel{\text{3a}}{}$ $\stackrel{\text{3a}}{}$}\end{gathered}\right)+P_n\left(\begin{gathered}\text{\tiny \usefont{U}{dice3d}{m}{n}$\stackrel{\text{3b}}{}$ 
$\stackrel{\text{3b}}{}$}\end{gathered}\right)\Big)$ 
then $w_2=w_3=\dots=w_{\infty}=0$. The linear function of event probabilities in \eqref{eq:boundP2} acts as a ``witness''. Specifically, it can be used (by solving the minimisation problem) to distinguish the probability distribution of a finitely exchangeable sequence from that of an infinitely exchangeable one.
\end{example}

~\\
Why are these bounds important? First, statistical problems and experiments involve exchangeable sequences that are always finite. It is not possible to perform an infinite number of trials. Second,  de Finetti's theorem provides a useful and simple representation of the mechanism generating the data (which, in turn, can be inverted to estimate $\boldsymbol{\theta}$ from the counts $n_1,\dots,n_6$ via Bayes' rule).
Consequently, practitioners are still inclined to use de Finetti's theorem as an approximation for the mechanism generating the finite exchangeable sequence and use these bounds to assess the error of the approximation for a particular inference.\footnote{In the previous example, the inference is the computation of $P_n\left(\begin{gathered}\text{\tiny \usefont{U}{dice3d}{m}{n}$\stackrel{\text{3a}}{}$ $\stackrel{\text{3a}}{}$}\end{gathered}\right)- P_n\left(\begin{gathered}\text{\tiny \usefont{U}{dice3d}{m}{n}$\stackrel{\text{3a}}{}$ 
$\stackrel{\text{3b}}{}$
}\end{gathered}\right)+P_n\left(\begin{gathered}\text{\tiny \usefont{U}{dice3d}{m}{n}$\stackrel{\text{3b}}{}$ 
$\stackrel{\text{3b}}{}$}\end{gathered}\right)$.} These facts motivate the search for representation theorems that apply to both the finitely and the infinitely exchangeable case.

\section{First representation}
\label{sec:signed}
A fact about finite exchangeability that has not received equal attention is that a de Finetti-like representation actually holds in the finite case if we drop the non-negativity condition for $q$. The result was already known to de Finetti \cite{finetti1930}. Explicitly derived in \cite{dellacherie1982probabilities,jaynes1982some} for some special cases, it was generalised by \cite{kerns2006definetti}.

\begin{proposition}[\cite{kerns2006definetti}]
\label{prop:kerns}
Given a  sequence of finitely exchangeable variables $t_1,t_2,\dots,t_r$, there exists a signed measure  $\nu$, satisfying  $\int_{\Theta} d\nu(\boldsymbol{\theta})=1$, such that:
 \begin{equation}
\label{eq:reprThsigned0}
P(t_1,\dots,t_r) = \int_\Theta \theta_1^{n_1}\theta_2^{n_2}\cdots \theta_5^{n_5}\theta_6^{n_6} d\nu({\boldsymbol \theta}).
\end{equation}
\end{proposition}
Signed measure  means that $\nu$  can have  negative values.   Note that, in  Eq.~\eqref{eq:reprThsigned0}, the only valid  signed-measures $\nu$ are those which give rise to well-defined (non-negative and normalised) probabilities $P(t_1,t_2,\dots,t_r)$.  

By comparing the two representations \eqref{eq:reprTh0} and \eqref{eq:reprThsigned0}, it becomes apparent that there exists a significant difference that may not be immediately noticeable.

Any probability measure $q$ originates, via Eq.~\eqref{eq:reprTh0}, into a valid (non-negative and normalised) probability  $P(t_1,\dots,t_r)$, that is for any $[n_1,\dots,n_6]$ such that $\sum_{i=1}^6 n_i=r$, we have that:
{
\begin{align}
 \label{eq:sum1neg0}
 &\int_\Theta \theta_1^{n_1}\theta_2^{n_2}\cdots \theta_5^{n_5}\theta_6^{n_6} dq({\boldsymbol \theta}) \geq 0, ~\\
 \label{eq:sum1neg}
\sum_{[n'_1,\dots,n'_6]: \sum_{i=1}^6 n'_i=r} b_{n'_1n'_2\dots n'_6}&\int_\Theta \theta_1^{n'_1}\theta_2^{n'_2}\cdots \theta_5^{n'_5}\theta_6^{n'_6} dq({\boldsymbol \theta})=1,
\end{align}
where $b_{n_1n_2\dots n_6}=r!/(n_1!n_2!n_3!n_4!n_5!n_6!)$ is the multinomial coefficient. Note that, \eqref{eq:sum1neg0}--\eqref{eq:sum1neg} imply that $P(t_1,\dots,t_r)\geq0$ and $\sum_{t_1,\dots,t_r} P(t_1,\dots,t_r)=1$. Indeed, consider  $P(t_1,\dots,t_r)$$
$ and let  $n_i=\sum_{j=1}^r \delta_{i}(t_j)$  the number of occurrences of the i-th face in the sequence $t_1,\dots,t_r$.
From \eqref{eq:reprTh0}, we have that \\ $
P(t_{1},\dots,t_{r}) = \int_\Theta \theta_1^{n_1}\theta_2^{n_2}\cdots \theta_5^{n_5}\theta_6^{n_6} dq({\boldsymbol \theta})$ (which is nonnegative). This equality holds for any permutation $t_{\pi(1)},t_{\pi(2)},\dots,t_{\pi(r)}$, as they all have the same amount of counts $n_1,n_2,\dots,n_6$. There are $r!/(n_1!n_2!n_3!n_4!n_5!n_6!)$ of these permutations. Therefore, we have that $$\sum_{t_1,\dots,t_r}P(t_1,\dots,t_r)=\sum\limits_{{\bf n}= [n'_1,\dots,n'_6]: \sum_i n'_i=r} \hspace{-0.5mm}b_{n'_1n'_2\dots n'_6} 
\int_\Theta \theta_1^{n'_1}\theta_2^{n'_2}\cdots \theta_5^{n'_5}\theta_6^{n'_6} dq({\boldsymbol \theta})=1
$$ (normalised). 
} 
Notice that not all signed measures $\nu$ that integrate to one define valid probabilities  $P(t_1,\dots,t_r)$ via Eq.~\eqref{eq:reprThsigned0}. In some cases, these measures may result in a negative value for $P(t_1,\dots,t_r)$. This is why this representation is not used in practice. Similarly, this is why in QT we avoid using `negative probabilities', despite the well-known fact that such probabilities could enable a hidden variable representation, at least mathematically.

For finite exchangeability, 
Jaynes' work \cite{jaynes1982some} is interesting because it contains a  method to define valid $P$ via Eq.~\eqref{eq:reprThsigned0} using Bernstein (or Legendre) polynomials. A more recent approach is discussed in the next section.

\section{Second representation}
\label{sec:Bern}
Reasoning about finitely exchangeable variables $t_i$ can be reduced to reasoning about  polynomials of count vectors, known as Bernstein polynomials \cite{jaynes1982some,kerns2006definetti,harper2007probability,de2012exchangeability,zaffalon2019b}. Working with this polynomial representation automatically guarantees
that finite exchangeability is satisfied and that $P$ are valid probabilities, without having to deal directly with the signed measure $\nu$.

The \textit{multivariate Bernstein polynomials} are  polynomials of the form
\begin{equation}
    \label{eq:Berne}
\theta_1^{n_1}\theta_2^{n_2}\cdots \theta_5^{n_5}\left(1-\theta_1-\dots-\theta_5\right)^{n_6} \text{ with }\sum_{i=1}^6 n_i=r,
\end{equation}
where $r$ is the number of rolls. They satisfy many useful properties, such as forming a base for the linear space of all polynomials whose degree is at most  $r$ and forming a partition of unity. 
To simplify the notation, we will denote $1-\theta_1-\dots-\theta_5$ as $\theta_6$, so that $\sum_{i=1}^6 \theta_i=1$.

The fundamental concept behind the new representation theorem that  focuses on Bernstein polynomials involves working directly with the linear operator $L$ defined as
\begin{equation}
\label{eq:quasiexp}
    L(\dots) := \int_{\Theta} (\dots)\,d\nu(\boldsymbol{\theta}),
\end{equation}
and determining the specific properties that $L$ must adhere to in order to produce valid probabilities. %

More precisely, let $L$ %
be  a \textit{linear operator} from the space $\mathcal{B}_r$ of  polynomials
of  degree at most $r$ to $\mathbb{R}$. Our goal is to write any finitely exchangeable $P(t_1,t_2,\dots,t_r)$  as $L(\theta_1^{n_1}\theta_2^{n_2}\cdots \theta_6^{n_6})$, where $\sum_{i=1}^6 n_i=r$. In the aim of determining the appropriate properties of $L$, %
we first introduce the  closed convex cone %
\begin{equation}
 \label{eq:LPcone1}
\begin{aligned}
    \mathcal{B}^+_r =\Bigg\{ & \sum\limits_{{\bf n}= [n_1,\dots,n_6]: \sum_i n_i=r} \hspace{-3mm}u_{{\bf n}}\, 
  \theta_{1}^{n_1} \theta_{2}^{n_2}\cdots  \theta_6^{n_{6}}: u_{{\bf n}}\in \mathbb{R}^{+}  \Bigg\},
\end{aligned}
\end{equation}
which includes  all non-negative linear combinations of Bernstein's monomials of degree $r$, 
 as well as all constant polynomials with value $c\geq 0$. This latter property holds because %
{
\begin{equation}
    \label{eq:normalis}
\begin{aligned}&c=c(\theta_{1}+\theta_{2}+\cdots+\theta_6 )^{r}=c\sum\limits_{{\bf n}= [n_1,\dots,n_6]: \sum_i n_i=r} \hspace{-3mm}
b_{n_1n_2\dots n_6}\theta_{1}^{n_1} \theta_{2}^{n_2}\cdots  \theta_6^{n_{6}},
\end{aligned}
\end{equation}
and the coefficients $c\, b_{n_1n_2\dots n_6}$ are positive.}
  In the following, we will use $g$ to denote generic polynomials of the vector of variables $\boldsymbol{\theta}$ and $c$ to represent the constant polynomial with value $c$, meaning $g(\boldsymbol{\theta}) = c$ for all $\boldsymbol{\theta}$.\\

\begin{proposition}[\cite{zaffalon2019b}]
\label{prop:thBern}
Given a  sequence of finitely exchangeable variables $t_1,t_2,\dots,t_r$, there exists a linear operator $L: \mathcal{B}_r \rightarrow \mathbb{R}$  satisfying  
\begin{equation}
\label{eq:A}
     L(g)\geq \sup c \text{ s.t. } g-c \in \mathcal{B}_r^+,  
\end{equation}
for all polynomials $g \in \mathcal{B}_r$ and where $c$ is the constant polynomial of value $c$, such that:
 \begin{equation}
\label{eq:reprThsigned}
P(t_1,\dots,t_r) = L(\theta_1^{n_1}\theta_2^{n_2}\cdots \theta_5^{n_5}\theta_6^{n_6}),
\end{equation}
where $n_i=\sum_{j=1}^r \delta_{i}(t_j)$ denotes the number of occurrences of the i-th face in the sequence $t_1,\dots,t_r$.\\
\end{proposition}
{ With the expression in Eq. \eqref{eq:A}, we mean that  $L(g)\geq \gamma$ where $\gamma=\{\sup c \mid g-c \in \mathcal{B}_r^+\}$.}
Note that, the supremum in \eqref{eq:A} is taken over the value $c$, which defines the constant polynomial $c$.
We call  a  linear operator $L$ satisfying \eqref{eq:A}
a \emph{quasi-expectation operator}. The general definition of \emph{quasi-expectation operator}, along with some of its main properties, are outlined in  \ref{app:qexp}. Proposition \ref{prop:thBern} states that that  any probability distribution $P$ of an exchangeable sequence $t_1,\dots,t_r$  can be defined as a quasi-expectation of a polynomial of the variables $\theta_1,\dots,\theta_6$. Since the polynomial $\theta_1^{n_1}\theta_2^{n_2}\cdots \theta_5^{n_5}\theta_6^{n_6}$ only depends on the sequence $t_1,\dots,t_r$ through the number of occurrences  $n_i$, the probability $P$ defined by \eqref{eq:reprThsigned} is by definition finitely exchangeable.
The condition \eqref{eq:A}  ensures that the quasi-expectation operator $L$ defines valid (non-negative and normalised) probabilities  $P$ via \eqref{eq:reprThsigned}. In fact, note that, for $g=c$ (the constant polynomial), the condition \eqref{eq:A} is equal to $L(c)\geq c$. Similarly for $g=-c$, we have that $L(-c)\geq -c$. Since $L$ is linear, we have that $L(-c)=-L(c)\geq -c$ and so $L(c)\leq c$. The two inequalities imply that $L(c)=c$.  For $c=1$, we have that $L(1)=1$. Hence, from \eqref{eq:normalis} for $c=1$ and  by linearity of $L$, we can derive that
{
 $$
\begin{aligned}
\sum\limits_{{\bf n}= [n_1,\dots,n_6]: \sum_i n_i=r} \hspace{-3mm}
b_{n_1n_2\dots n_6} L\left(\theta_{1}^{n_1} \theta_{2}^{n_2}\cdots  \theta_6^{n_{6}}\right)=L(1)=1.
\end{aligned}
$$
Consider now $P(t_1,\dots,t_r)$$
$ and let  $n_i=\sum_{j=1}^r \delta_{i}(t_j)$  the number of occurrences of the i-th face in the sequence $t_1,\dots,t_r$.
From \eqref{eq:reprThsigned}, we have  that $P(t_1,\dots,t_r) = L(\theta_1^{n_1}\theta_2^{n_2}\cdots \theta_5^{n_5}\theta_6^{n_6})$. This equality holds for any permutation $t_{\pi(1)},t_{\pi(2)},\dots,t_{\pi(r)}$ as they all have the same amount of counts $n_1,n_2,\dots,n_6$. There are $r!/(n_1!n_2!n_3!n_4!n_5!n_6!)$ of these permutations. Therefore, we have that
$$\sum_{t_1,\dots,t_r}P(t_1,\dots,t_r)=\sum\limits_{{\bf n}= [n_1,\dots,n_6]: \sum_i n_i=r} \hspace{-0.5mm}b_{n_1n_2\dots n_6} L\left(\theta_{1}^{n_1} \theta_{2}^{n_2}\cdots  \theta_6^{n_{6}}\right)=L(1)=1
$$ (normalised). 
}
It remains to show that $P(t_1,\dots,t_r)>0$ for every $t_1,\dots,t_r$. Note that, for $g(\boldsymbol{\theta})=\theta_1^{n_1}\theta_2^{n_2}\cdots \theta_5^{n_5}\theta_6^{n_6}$, the supremum value of $c$ such that the polynomial $g-c \in \mathcal{B}_r^+$ is $c=0$. From \eqref{eq:A}, it follows  that $L(g)\geq0$ (non-negative) for every $g(\boldsymbol{\theta})=\theta_1^{n_1}\theta_2^{n_2}\cdots \theta_5^{n_5}\theta_6^{n_6}$.\\

\begin{remark}
It is important to note that $L$ is defined on 
$\mathcal{B}_r$, which is the space of polynomials of degree up to $r$, but since Bernstein polynomials of degree $r$ form a basis for the space of polynomials of degree up to $r$, $L$ is completely determined once we specify its value 
$L(g)$ for every Bernstein polynomial $g$ of degree $r$.\\
\end{remark}

\begin{example}
\label{eq:LP0} Consider the case $r=2$ and a probability $P$ defined via a linear operator $L$ such that $L(\theta_1\theta_2)=0.5$ and $L(\theta_i\theta_j)=0$ for any other monomial of degree two in $\mathcal{B}_2$. It is easy to check that $L$ satisfies \eqref{eq:A} and, therefore, defines a non-negative and normalised finitely exchangeable $P$.  From \eqref{eq:reprThsigned}, it follows that:
$$
\begin{aligned}
P(t_1=1,t_2=1)&=L(\theta_1^2)=0,\\
P(t_1=1,t_2=2)&=P(t_1=2,t_2=1)=L(\theta_1\theta_2)=0.5,\\
P(t_1=2,t_2=2)&=L(\theta_2^2)=0,
\end{aligned}
$$
and zero otherwise. 
\end{example}

~\\
We will now use this representation through the quasi-expectation operator to determine bounds for linear functions of event probabilities.  In particular, in the next example, we will exploit \eqref{eq:A} to compute a lower bound for $L(g)$ for $g(\boldsymbol{\theta})=\theta_1^2-\theta_1\theta_2+\theta^2$.
~\\

\begin{example}
\label{ex:LP}
Let us consider again the problem discussed in Example \ref{eq:diffP}, with $r=2$. Assume we aim to find the finitely exchangeable probability $P_2$ which minimises 
{\small
$$
v_2=\min_{P_2} \Bigg( P_2\left(\begin{gathered}\text{\tiny \usefont{U}{dice3d}{m}{n}$\stackrel{\text{3a}}{}$ $\stackrel{\text{3a}}{}$}\end{gathered}\right)- P_2\left(\begin{gathered}\text{\tiny \usefont{U}{dice3d}{m}{n}$\stackrel{\text{3a}}{}$ 
 $\stackrel{\text{3b}}{}$
}\end{gathered}\right)+P_2\left(\begin{gathered}\text{\tiny \usefont{U}{dice3d}{m}{n}$\stackrel{\text{3b}}{}$ 
 $\stackrel{\text{3b}}{}$}\end{gathered}\right)\Bigg).
$$
}
Using the operator $L$, this quantity can be rewritten as
$$
\begin{aligned}
v_2&=\min_{L_2} L_2(\theta_1^2-\theta_1\theta_2+\theta_2^2)\\
&=\min_{L_2} \Big(L_2(\theta_1^2)-L_2(\theta_1\theta_2)+L_2(\theta_2^2)\Big)
\end{aligned}
$$
where  $L_2$ indicates  the linear operator corresponding to $P_2$.

By \eqref{eq:A}, knowing that $\theta_6:=1-\theta_1-\dots-\theta_5$, the value $v_2$ can be determined by solving the following optimisation problem:\footnote{The supremum can be attained in this case and, therefore, we have replaced $\sup$ with $\max$.}
\begin{equation}
\label{eq:optimdice}
\begin{aligned}
v_2=&\max\, c
~~~s.t.~~~\theta_1^2-\theta_1\theta_2+\theta_2^2-c \in \mathcal{B}^+_2,
\end{aligned}
\end{equation}
where
\begin{equation}
\label{eq:coneposex}
\begin{aligned}
&\mathcal{B}^+_2=\{u_{200000}\theta_1^2 +  u_{110000}\theta_1\theta_2+  u_{101000}\theta_1\theta_3 + \\
& u_{100100}\theta_1\theta_4 +  u_{100010}\theta_1\theta_5\\
&+  u_{100001}\theta_1(1-\theta_1-\dots-\theta_5)  \\
&+u_{020000}\theta_2^2 +  \dots+ u_{000011}\theta_5(1-\theta_1-\dots-\theta_5)\\
&+ u_{000002}(1-\theta_1-\dots-\theta_5)^2: u_{[n_1,n_2,\dots,n_6]}\geq0\}.
\end{aligned}
\end{equation}
For a particular $c$, the polynomial $\theta_1^2-\theta_1\theta_2+\theta_2^2-c$ belongs to $\mathcal{B}^+_2$ if we can find coefficients $u_{200000},u_{110000},\dots,u_{000002}\geq0$ such that 

\begin{equation}
\label{eq:constrex}
\begin{aligned}
&\theta_1^2-\theta_1\theta_2+\theta_2^2-c =u_{200000}\theta_1^2 +  u_{110000}\theta_1\theta_2\\
&+  u_{101000}\theta_1\theta_3 +  u_{100100}\theta_1\theta_4 +  u_{100010}\theta_1\theta_5\\
&+  u_{100001}\theta_1(1-\theta_1-\dots-\theta_5) \\
&+u_{020000}\theta_2^2 +  \dots+ u_{000011}\theta_5(1-\theta_1-\dots-\theta_5)\\
&+ u_{000002}(1-\theta_1-\dots-\theta_5)^2.
\end{aligned}
\end{equation}

Therefore, the solution of \eqref{eq:optimdice} can be computed
 by solving the following linear programming problem:
\begin{equation}
\label{eq:constrU}
\begin{aligned}
&\max_{c\in \mathbb{R},u_{[n_1,n_2,\dots,n_6]}\in \mathbb{R}^+} c  \\
&s.t.\\
& -u_{000002} + u_{100001} - u_{200000} + 1 =0\\
& -2u_{000002} + u_{010001} + u_{100001} - u_{110000} - 1=0\\
 &-u_{000002} + u_{010001} - u_{020000} + 1=0\\
 &- c - u_{000002}=0\\
 & -2u_{000002} -u_{101000}+u_{100001}+ u_{001001}=0\\
 &  ~~~\texttt{ additional 21 constraints  }
\end{aligned}
\end{equation}
where the equality constraints have been obtained by equating the coefficients of the monomials
in \eqref{eq:constrex}.
For instance, consider the  constant term in the r.h.s.\ of \eqref{eq:constrex}: that is $u_{000002}$. It must be equal to the constant term $-c$. Therefore, we have that
$$
-c-u_{000002}=0.
$$
Similarly, consider the coefficient of the monomial $\theta_1^2$ in the r.h.s.\ of  \eqref{eq:constrex}: this is $u_{200000}-u_{100001}+u_{000002}$. It must be equal to the coefficient of the monomial $\theta_1^2$ in $\theta_1^2-\theta_1\theta_2+\theta_2^2-c$, which is $1$, that is
$$
1-(u_{200000}-u_{100001}+u_{000002})=0.
$$
The last specified constraint in Eq.~\eqref{eq:constrU} is relative to the monomial $\theta_1\theta_3$: since $\theta_1\theta_3$ is not  in $\theta_1^2-\theta_1\theta_2+\theta_2^2-c$, its coefficients are equated to zero. In total, there are 26 constraints.

The solution of Eq.~\eqref{eq:optimdice} can be computed using any linear programming solver  and is
$$
v_2=-0.5.
$$
The problem 
{\small
$$
v_3=\min_{P_3} \Bigg(
P_3\left(\begin{gathered}\text{\tiny \usefont{U}{dice3d}{m}{n}$\stackrel{\text{3a}}{}$ $\stackrel{\text{3a}}{}$}\end{gathered}\right)- P_3\left(\begin{gathered}\text{\tiny \usefont{U}{dice3d}{m}{n}$\stackrel{\text{3a}}{}$ 
 $\stackrel{\text{3b}}{}$
}\end{gathered}\right)+P_3\left(\begin{gathered}\text{\tiny \usefont{U}{dice3d}{m}{n}$\stackrel{\text{3b}}{}$ 
 $\stackrel{\text{3b}}{}$}\end{gathered}\right)\bigg),
$$
}
\\ where $P_3(t_1=a,t_2=b)=\sum_{i=1}^6P _3(t_1=a,t_2=b,t_3=i)$ and $P_3$ is the probability distribution of the finitely exchangeable sequence $t_1,t_2,t_3$, can be also solved via linear programming. In this case, we only need to express the above marginalisation in terms of the $L_3$ operator, for instance:
\begin{equation}
\label{eq:partialtrex}
\begin{aligned}
&P_3\left(\begin{gathered}\text{\tiny \usefont{U}{dice3d}{m}{n}$\stackrel{\text{3a}}{}$ $\stackrel{\text{3a}}{}$}\end{gathered}\right)=L_3\big(\theta_1^2(\theta_1+\theta_2+\dots+\theta_6)\big)\\
&P_3\left(\begin{gathered}\text{\tiny \usefont{U}{dice3d}{m}{n}$\stackrel{\text{3a}}{}$ $\stackrel{\text{3b}}{}$}\end{gathered}\right)=L_3\big(\theta_1\theta_2(\theta_1+\theta_2+\dots+\theta_6)\big)\\
&P_3\left(\begin{gathered}\text{\tiny \usefont{U}{dice3d}{m}{n}$\stackrel{\text{3b}}{}$ $\stackrel{\text{3b}}{}$}\end{gathered}\right)=L_3\big(\theta_2^2(\theta_1+\theta_2+\dots+\theta_6)\big)
\end{aligned}
\end{equation}
\\ where we have exploited that $\theta_1+\theta_2+\dots+\theta_6=1$. We can thus compute marginals on the result of two rolls by using a linear operator which is only defined on polynomials of degree 3, an operation that, %
as clarified in the next section, corresponds to the operation  \textbf{partial trace} in QT.
The solution is $$
v_3\approx -0.167.
$$

Note that, in the infinitely exchangeable case, the operator $L_{\infty}(\dots) := \int_{\Theta} (\dots)\,d q(\boldsymbol{\theta})$    applied to the polynomial $\theta_1^2-\theta_1\theta_2+\theta_2^2$, satisfies 
$$
\begin{aligned}
v_\infty&=\inf_{L_{\infty}} L_{\infty}(\theta_1^2-\theta_1\theta_2+\theta_2^2)\\&=\inf_{L_{\infty}} L_{\infty}(\theta_1^2)-L_{\infty}(\theta_1\theta_2)+L_{\infty}(\theta_2^2)=0.  
\end{aligned}
$$
Observe that, $v_\infty=0$ because $\theta_1^2-\theta_1\theta_2+\theta_2^2$ is nonnegative and its infimum is $0$.

\begin{figure}[h]
    \centering
    \includegraphics[width=7cm]{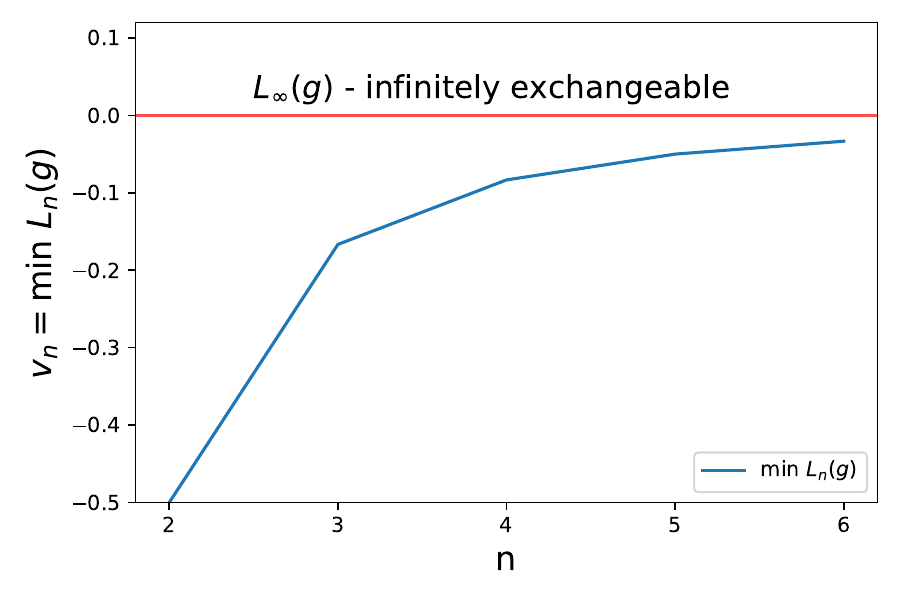}
    \caption{Value of $v_n$ as a function of $n$. The red line corresponds to $v_{\infty}$.}
    \label{fig:convLP}
\end{figure}

Figure \ref{fig:convLP} shows the value of $v_n=\min L_n(g)$ for $g(\boldsymbol{\theta})=\theta_1^2-\theta_1\theta_2+\theta_2^2$ and $n=2,\dots,6$,  where the subscript in $L_n$ indicates the linear operator corresponding to $P_n$, that is
{\small
$$
\begin{aligned}
v_n&=\min_{P_n} \Bigg( P_n\left(\begin{gathered}\text{\tiny \usefont{U}{dice3d}{m}{n}$\stackrel{\text{3a}}{}$ $\stackrel{\text{3a}}{}$}\end{gathered}\right)- P_n\left(\begin{gathered}\text{\tiny \usefont{U}{dice3d}{m}{n}$\stackrel{\text{3a}}{}$ 
 $\stackrel{\text{3b}}{}$
}\end{gathered}\right)+P_n\left(\begin{gathered}\text{\tiny \usefont{U}{dice3d}{m}{n}$\stackrel{\text{3b}}{}$ 
 $\stackrel{\text{3b}}{}$}\end{gathered}\right)\Bigg)
\end{aligned}
$$
}
\\ as a function of $n$. Notice that the gap between the value computed using de Finetti's theorem (red line) and the value computed for finitely exchangeable probabilities reduces as of the number of rolls $n$ increases. This implies that, as the total number of rolls increases, the worst-case marginal expectation for two rolls, which are part of a longer finitely exchangeable sequence of additional $0,1,2,\dots$ rolls, converges to the classical  probability (infinitely exchangeable) value. We can refer to this phenomenon as the `emergence of classical probability'. We will discuss it again  in the next Section.

\begin{figure}[h]
    \centering
    \includegraphics[width=9cm]{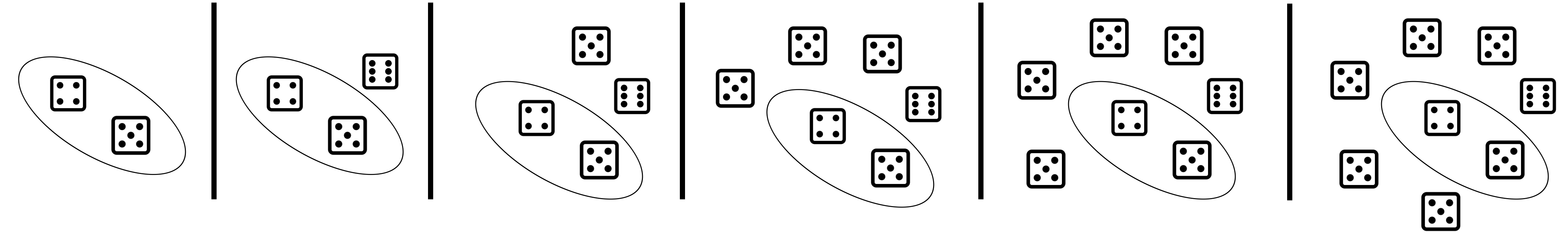}
    \caption{The worst-case marginal expectation for two rolls, which are part of a longer finitely exchangeable sequence of additional $0,1,2,\dots$ rolls, converges to the classical 
    probability value.}
    \label{fig:dicecircle}
\end{figure}
\end{example}

\section{Making the connection with QT}
\label{sec:Ent}
In this section, we will  strengthen the connection between 
finite exchangeability and  finite dimensional QT.

Let ${\bf x}$ be a complex vector in $\mathbb{C}^d$ such that ${\bf x}^\dagger {\bf x}=1$. We can express the probabilities of the six faces ($d=6$) of the dice as:
$$
{\boldsymbol \theta}=[\theta_1,\dots,\theta_6]=[x_1^\dagger x_1,\dots,x_6^\dagger x_6].
$$
where $^\dagger$ denotes conjugate transpose. 
Moreover, notice that we can represent the monomials of the Bernstein polynomials with a specific degree $r$ as $\otimes_r {\boldsymbol \theta}$, where $\otimes$ is the Kronecker product. It follows that any polynomial in $\mathcal{B}_r$ can equivalently be written as 
\begin{equation}
\label{eq:ginx}
   g({\bf \theta})=(\otimes_r {\bf x})^\dagger D (\otimes_r {\bf x}),
\end{equation}
{ for some  diagonal Hermitian matrix $D$}.\\
\begin{example}
\label{ex:thetax}
The polynomial $\theta_1^2-\theta_1\theta_2+\theta_2^2$ can be written as   
$$
\begin{aligned}
g({\bf x}^\dagger,{\bf x})&=(x_1^\dagger x_1)^2-(x_1^\dagger x_1)(x_2^\dagger x_2)+(x_2^\dagger x_2)^2 \\
&=(\otimes_2 {\bf x})^\dagger D (\otimes_2 {\bf x})
\end{aligned}
$$
where $D$ is a diagonal matrix {($\text{diag}(D)=[D_1,D_2,\dots,D_{36}]^\top$)} with all zero entries expect $D_{1}=1=D_{8}$ and $D_{2}=D_{7}=-0.5$. Note that, to emphasise that the polynomials we will consider in this section are function of a complex vector $\mathbf{x}$ and its conjugate-transpose ${\bf x}^\dagger$, we will denote them as $g({\bf x}^\dagger,{\bf x})$.\\
\end{example}
We now
introduce a new linear operator $\widehat{L}:\mathcal{S}_r \rightarrow \mathbb{R}$, where 
$$
\mathcal{S}_r:=\{(\otimes_r {\bf x})^\dagger G (\otimes_r {\bf x}): G \text{ is Hermitian}\},
$$
is  the vector space of  polynomials of degree $2r$ of the variables $x_i,x_i^\dagger$ for $i=1,\dots,d$. Note that, $\mathcal{S}_r$ includes the polynomials of the form \eqref{eq:ginx}.

From  linearity  of the trace $Tr(\cdot)$ and of $\widehat{L}$, we have: 
\begin{equation}
\begin{aligned}
   \widehat{L}(g)&=\widehat{L}((\otimes_r {\bf x})^\dagger D (\otimes_r {\bf x}))\\
   &= Tr\Big(D\widehat{L}\Big( (\otimes_r {\bf x})(\otimes_r {\bf x})^\dagger\Big)\Big).
   \end{aligned}
\end{equation}
Note that, $\widehat{L}\Big( (\otimes_r {\bf x})(\otimes_r {\bf x})^\dagger\Big)$ is a Hermitian  matrix of dimension $d^r \times d^r$. Example for $r=2,d=6$:
\begin{align}
\widehat{L}\Big( (\otimes_2 {\bf x})(\otimes_2 {\bf x})^\dagger\Big)=
&\widehat{L}\left(\begin{matrix}
(x_1^2)^\dagger x_1^2& (x_1^2)^\dagger x_1 x_2 & \dots &(x_1^2)^\dagger x_6^2\\
(x_1 x_2)^\dagger x_1^2 & (x_1 x_2)^\dagger x_1 x_2 & \dots &(x_1 x_2)^\dagger x_6^2\\
(x_1 x_3)^\dagger x_1^2 & (x_1 x_3)^\dagger x_1 x_2 & \dots &(x_1 x_3)^\dagger x_6^2\\
\vdots & \vdots & \vdots  & \vdots\\
(x_6^2)^\dagger x_1^2 & (x_6^2)^\dagger x_1 x_2 & \dots &(x_6^2)^\dagger x_6^2\\
\end{matrix}\right),
\end{align}
where $\widehat{L}$ is applied element-wise.  Observe that the above matrix has additional symmetries (besides being Hermitian), which follow by the symmetries of $(\otimes_r {\bf x})(\otimes_r {\bf x})^\dagger$. We can then prove the following representation theorem.\\

\begin{theorem}
\label{th:repr41}
Consider a linear operator $\widehat{L}: \mathcal{S}_r \rightarrow \mathbb{R}$ satisfying  
\begin{equation}
\label{eq:A1}
     \widehat{L}(g)\geq \sup c \text{ s.t. } g-c \in \mathcal{S}_r^+,  
\end{equation}
for all polynomials $g \in \mathcal{S}_r$, where $c$ is the constant polynomial of value\footnote{This can be obtained by choosing $H=cI$ where $I$ is the identity matrix.} $c$
and $\mathcal{S}_r^+=\{(\otimes_r {\bf x})^\dagger H (\otimes_r {\bf x}):\text{H is Hermitian and PSD}\}$ with PSD meaning Positive Semi-Definite.
We can then prove that
{
\begin{enumerate}
    \item If $\widehat{L}$ satisfies \eqref{eq:A1} for each polynomial $g\in \mathcal{S}_r$, then  the matrix  $M=\widehat{L}\Big( (\otimes_r {\bf x})(\otimes_r {\bf x})^\dagger\Big)$ is PSD with trace one. %
\item $P(t_1,\dots,t_r)=\widehat{L}((\otimes_r {\bf x})^\dagger D^{(t_{1:r})} (\otimes_r {\bf x}))=Tr(D^{(t_{1:r})} M)$ --- where $D^{(t_{1:r})}$ is a diagonal Hermitian matrix depending on $t_1,\dots,t_r$,\\ and $M=\widehat{L}\Big( (\otimes_r {\bf x})(\otimes_r {\bf x})^\dagger\Big)$ --- is a probability (nonnegative and normalised) and  is symmetric to permutations of the labels of the variables $t_i$.
\end{enumerate}}
\end{theorem}
We have constructed a quasi-expectation operator $\widehat{L}$ on a space of polynomials of complex variables, which is completely determined by trace-one PSD matrices $M$ (PSD with trace one and having  the same symmetries as the matrix $(\otimes_r {\bf x})(\otimes_r {\bf x})^\dagger$). Moreover, this operator, when applied to polynomials $(\otimes_r {\bf x})^\dagger D^{(t_{1:r})} (\otimes_r {\bf x})$, defines a probability $P$ which is symmetric to permutations of
the labels of the variables $t_i$. 
{
\begin{example}
Consider the case   $d=2$ (a coin) and $r=3$.
Then, for instance, we have that
$$
\begin{aligned}
p(t_1=1,t_2=1,t_3=2)&=(\otimes_3 {\bf x})^\dagger D^{(112)} (\otimes_3 {\bf x})=(x_1^\dagger x_1)^2 (x_2^\dagger x_2),\\
p(t_1=1,t_2=2,t_3=1)&=(\otimes_3 {\bf x})^\dagger D^{(121)} (\otimes_3 {\bf x})=(x_1^\dagger x_1)^2 (x_2^\dagger x_2),\\
p(t_1=2,t_2=1,t_3=1)&=(\otimes_3 {\bf x})^\dagger D^{(211)} (\otimes_3 {\bf x})=(x_1^\dagger x_1)^2 (x_2^\dagger x_2),\\
\end{aligned}
$$
with ${\bf x}=[x_1,x_2]^\top$ and $D^{(112)} =\text{diag}(0,1,0,0,0,0,0,0)$, \\$D^{(121)} =\text{diag}(0,0,1,0,0,0,0,0)$, $D^{(211)} =\text{diag}(0,0,0,0,1,0,0,0)$.
\end{example}
}

For the second statement in Theorem \ref{th:repr41}, note that $Tr(DM)=Tr(D\,\mathrm{diag}(M))$ holds for any diagonal matrix $D$. Hence Theorem \ref{th:repr41} uses some redundant quantities through $\widehat{L}\Big( (\otimes_r {\bf x})(\otimes_r {\bf x})^\dagger\Big)$. This withstanding, it enables another representation result that links finite exchangeability to QT, as discussed next.

First, we recall some useful definitions. Let $\mathcal{H}^{\otimes r}$ be the compound Hilbert space of $r$ particles. A fundamental postulate of QT for
bosons  is that the quantum states are symmetric  upon exchange of the subsystems, respectively. This postulate is formulated mathematically by introducing a projection operator  $\Pi_{sym}$ from
the tensor product space to this symmetric subspace.  The projector is called \textit{symmetriser} and is defined as   \cite[B-2-c]{cohen2020quantum}:
\begin{equation}\label{eq:symm}
 \Pi_{sym} := \frac{1}{N!} \sum_{i}\Pi_i,    
\end{equation}
where the summations are performed over the $N!$ permutation matrices, $\Pi_i$, associated to the permutations of the labels of the particles. Then the postulate states that a density matrix $\rho$ of $r$  particles is said to be \textit{Bose-symmetric} if:
$$
\rho = \Pi_{sym} \rho \Pi_{sym}.
$$
As briefly derived in  \ref{app:bosesymm}, this definition is equivalent to the more commonly employed one (for instance in \cite{hudson1976locally,christandl2007one}), namely, $
\rho = \Pi_{i} \rho = \rho \Pi_{i}^\dagger$ for every permutation matrix $\Pi_i$.
{
\begin{example}
Assume that $\text{dim}(\mathcal{H})=d=2$ and $r=2$. In this case, the symmetriser is  equal to:
$$
\Pi_{sym}=\begin{bmatrix}
1 & 0 & 0 &0\\
0 & 0.5 & 0.5 &0\\
0 & 0.5 & 0.5 &0\\
0 & 0 & 0 &1\\
\end{bmatrix}.
$$

It can  be verified that, for instance, the following density matrices
\begin{equation}
\label{eq:exrho}
\frac{1}{2}\begin{bmatrix}
1 & 0 & 0 &1\\
0 & 0 & 0 &0\\
0 & 0 & 0 &0\\
1 & 0 & 0 &1\\
\end{bmatrix},~~~~\frac{1}{2}\begin{bmatrix}
0 & 0 & 0 &0\\
0 & 1 & 1 &0\\
0 & 1 & 1 &0\\
0 & 0 & 0 &0\\
\end{bmatrix},
\end{equation}
are all Bose-symmetric, as they satisfy $\rho=\Pi_{sym} \rho \Pi_{sym}$. Let us write a general density matrix as:
\begin{equation}
\label{eq:genericrho}
\rho=\begin{bmatrix}
\rho_{11} & \rho_{12} & \rho_{13} & \rho_{14} \\[6pt]
\rho_{12}^\dagger & \rho_{22} & \rho_{23} & \rho_{24} \\[6pt]
\rho_{13}^\dagger & \rho_{23}^\dagger & \rho_{33} & \rho_{34} \\[6pt]
\rho_{14}^\dagger & \rho_{24}^\dagger & \rho_{34}^\dagger & \rho_{44}
\end{bmatrix},
\end{equation}
which must also be Hermitian with trace-one (the elements in the diagonal are nonnegative real numbers  and sum up to one), then
\begin{equation}
\label{eq:genericrhosym}
\Pi_{sym} \rho \Pi_{sym}=\left[
\begin{smallmatrix}
\rho_{11} & \tfrac{1}{2}(\rho_{12} + \rho_{13}) & \tfrac{1}{2}(\rho_{12} + \rho_{13}) & \rho_{14} \\[6pt]
\tfrac{1}{2}(\rho_{12}^\dagger + \rho_{13}^\dagger) &
\tfrac{\rho_{22}+\rho_{33} +\rho_{23} + \rho_{23}^\dagger}{4} &
\tfrac{\rho_{22}+\rho_{33} +\rho_{23} + \rho_{23}^\dagger}{4} &
\tfrac{1}{2}(\rho_{24} + \rho_{34}) \\[6pt]
\tfrac{1}{2}(\rho_{12}^\dagger + \rho_{13}^\dagger) &
\!\tfrac{\rho_{22}+\rho_{33} +\rho_{23} + \rho_{23}^\dagger}{4} &
\tfrac{\rho_{22}+\rho_{33} +\rho_{23} + \rho_{23}^\dagger}{4} &
\tfrac{1}{2}(\rho_{24} + \rho_{34}) \\[6pt]
\rho_{14}^\dagger &
\tfrac{1}{2}(\rho_{24}^\dagger + \rho_{34}^\dagger) &
\tfrac{1}{2}(\rho_{24}^\dagger + \rho_{34}^\dagger) &
\rho_{44}
\end{smallmatrix}\right].
\end{equation}
We can then define the conditions for a generic matrix $\rho$ to be Bose-symmetric by equating the elements of the two matrices in \eqref{eq:genericrho} and \eqref{eq:genericrhosym}.
For instance, for the element in the first-row and second-column, this leads to 
$\rho_{12}=\tfrac{1}{2}(\rho_{12} + \rho_{13})$, which implies that $\rho_{12}=\rho_{13}$. Similarly, we can derive that $\rho_{24}=\rho_{34}$. For the four central elements, we have that
$\rho_{22}=\rho_{23}=\rho_{33}=\tfrac{\rho_{22}+\rho_{33} +\rho_{23} + \rho_{23}^\dagger}{4}
$, which is a real number. Therefore, a generic $4 \times 4$ Bose-symmetric matrix must have these symmetries:
\begin{equation}
\label{eq:genericrhofin}
\rho=\begin{bmatrix}
\rho_{11} & \rho_{12} & \rho_{12} & \rho_{14} \\[6pt]
\rho_{12}^\dagger & \rho_{22} & \rho_{22} & \rho_{24} \\[6pt]
\rho_{12}^\dagger & \rho_{22} & \rho_{22} & \rho_{24} \\[6pt]
\rho_{14}^\dagger & \rho_{24}^\dagger & \rho_{24}^\dagger & \rho_{44}
\end{bmatrix},
\end{equation}
where $\rho_{ii} \in \mathbb{R}$ and  $\rho_{11}+2\rho_{22}+\rho_{44}=1$.
\end{example}
}
 We can in general  prove %
the following equivalence result.\\
\begin{theorem}
\label{th:2}
The following sets are equal:
\begin{align}
 \label{eq:states1}
\mathcal{S}_1&=\{\rho: ~\rho=\Pi_{sym}\rho\Pi_{sym}, ~\rho\succeq 0, ~Tr(\rho)=1\},\\
 \label{eq:states2}
\mathcal{S}_2&=\{M:~M=\widehat{L}( (\otimes_r {\bf x})(\otimes_r {\bf x})^\dagger)\},
\end{align}
where $\rho$ is the density matrix of $r$-indistinguishable bosons (each one being a $d$-level quantum system), and { $\rho \succeq 0$ means that the matrix $\rho$ is positive semi-definite}.\\
\end{theorem}

This results states that  boson-symmetric density matrices and the set of matrices $M$ obtained by applying the operator $\widehat{L}$ to the matrix $(\otimes_r {\bf x})(\otimes_r {\bf x})^\dagger$ of $r$-th degree polynomials in the variables $x_i,x_i^\dagger$ for $i=1,\dots,d$ are equivalent. Eq.~\eqref{eq:states2} is obtained by applying a simple change of variables to the classical probability models for finitely exchangeable probabilities, and actually offers an alternative representation of the second-quantisation formalism in QT, which is of interest in its own right (as discussed in the proof of Theorem \ref{th:2} in  \ref{app:thm2}). \\
{
\begin{example}
For the case $\text{dim}(\mathcal{H})=d=2$ and $r=2$, this equivalence can be seen by comparing the matrix  in \eqref{eq:genericrhofin} with \begin{align}
\widehat{L}\Big( (\otimes_2 {\bf x})(\otimes_2 {\bf x})^\dagger\Big)=
&\widehat{L}\left(\begin{matrix}
(x_1^2)^\dagger x_1^2& (x_1 x_2)^\dagger x_1^2 & (x_1 x_2)^\dagger x_1^2 & (x_2^2)^\dagger x_1^2\\
(x_1^2)^\dagger x_1 x_2& (x_1 x_2)^\dagger x_1 x_2 & (x_1 x_2)^\dagger x_1 x_2 & (x_2^2)^\dagger x_1 x_2\\
(x_1^2)^\dagger x_1 x_2& (x_1 x_2)^\dagger x_1 x_2 & (x_1 x_2)^\dagger x_1 x_2 & (x_2^2)^\dagger x_1 x_2\\
(x_1^2)^\dagger x_2^2& (x_1 x_2)^\dagger x_2^2 & (x_1 x_2)^\dagger x_2^2 & (x_2^2)^\dagger x_2^2\\
\end{matrix}\right),
\end{align}
computed for ${\bf x}=[x_1,x_2]^\top$. The two matrices have  the same symmetries. It can also be noted that, in both matrices, the four central elements are real numbers and the diagonal elements sum up to one, as $(x_1^\dagger x_1)^2+2(x_1 x_2)^\dagger x_1 x_2 +(x_2^\dagger x_2)^2=\theta_1^2+2\theta_1\theta_2+\theta_2^2=(\theta_1+\theta_2)^2=1$ and $\widehat{L}(1)=1$.
\end{example}
From Theorems \ref{th:repr41}--\ref{th:2}, the following result follows directly.
 }

\begin{corollary}
 The operator defined by  $\widehat{L}(g):=Tr(G M)$ for all $g \in \mathcal{S}_r$, with $M$ being (i) PSD; (ii) trace one; and (iii) having the same symmetries as $(\otimes_r {\bf x})(\otimes_r {\bf x})^\dagger$; satisfies \eqref{eq:A1}  for each $g \in \mathcal{S}_r$.   
\end{corollary}
\begin{remark}~\vspace{-0.2cm}\newline
\begin{itemize}
    \item In \cite{benavoli2021indi,BENAVOLI2022389}, we derived a representation for indistinguishable particles for bosons using exchangeability and the symmetries $\rho=\Pi_{sym}\rho\Pi_{sym}$, but not in the same way as the de Finetti's representation. In the de Finetti's representation of the probability distribution for an exchangeable sequence, the only relevant information about the sequences is the number of occurrences of an event (such as the number of rolls $n_i$ of the 
i-th face). This naturally leads to a sort of second-order quantisation representation, where the variables (e.g., \( \theta_1, \dots, \theta_6 \)) appear exponentiated (e.g., \( \theta_1^{n_1}, \dots, \theta_6^{n_6} \)) in the polynomial term, 
such as in  Eq.\ \eqref{eq:reprTh0}. In this paper, we have shown that we can provide a similar representation (with exponentiated variables $\mathbf{x},\mathbf{x}^\dagger$) for QT (bosons), which is equivalent to  $\rho=\Pi_{sym}\rho\Pi_{sym}$, as proven in Theorem \ref{th:2}.\\
    \item An approach based on a similar exponentiated polynomial representation was proposed in \cite{doherty2004complete} to define a criterion to determine if a density matrix of \textit{distinguishable} particles is entangled or not. This approach is known as the Doherty-Parrilo-Spedalieri (DPS) hierarchy. For a bipartite density matrix $\rho^{AB}$, the level-$k$ DPS relaxation determines whether there exists an extension $\tilde{\rho}^{A_1 \dots A_k B_1 \dots B_k}$ that satisfies the following conditions: (i) it is invariant under left or right multiplication by any permutation of the $A$ or $B$ subsystems, and (ii) it remains positive semidefinite under partial transposition of any subset of the systems. The latter condition is commonly referred to as \text{Positivity under Partial Transpose (PPT)}. Such a level-$k$ extension exists if and only if $\rho^{AB}$ is not entangled  \cite[Th.~1]{doherty2004complete}. Notably, finding such an extension $\tilde{\rho}$ reduces to solving $k$ semidefinite programming problems or, dually, $k$ sum-of-squares optimisation problems \cite{lasserre2009moments}.
Their main result \cite[Th.~1]{doherty2004complete} is proven using the quantum de Finetti theorem.  Unlike this method, which serves as a mathematical optimisation tool to determine whether a density matrix is not entangled, our main results (Theorem \ref{th:repr41} and \ref{th:2}) provides a representation theorem for boson-symmetric density matrices. However,  the condition \eqref{eq:A1} that defines  valid linear operators $\widehat{L}$ could also  be used to define a hierarchy of optimisation problems to verify if the density matrix of $n$ indistinguishable particles is not entangled, using \cite[Th.~1]{doherty2004complete}. 
\item Theorem \ref{th:2} demonstrates that any finite system of exchangeable events is equivalent to a quantum experiment involving indistinguishable bosons measured using projective measurements along the coordinate axes. This implies that we can build a  quantum experiment whose probability of the measurement outcomes coincides with the probability of a finitely exchangeable sequence of dice rolling.  This means that quantum experiments provide an alternative way to the urn-model  discussed in Example \ref{ex:urn}  to realise a finitely exchangeable probability.
To illustrate this observation, we provide an example.
    \end{itemize}
\end{remark}

\begin{example}
\label{ex:last}
We can compute   $v_n=\min_{P_n} P_n\left(\begin{gathered}\text{\tiny \usefont{U}{dice3d}{m}{n}$\stackrel{\text{3a}}{}$ $\stackrel{\text{3a}}{}$}\end{gathered}\right)- P_n\left(\begin{gathered}\text{\tiny \usefont{U}{dice3d}{m}{n}$\stackrel{\text{3a}}{}$ 
 $\stackrel{\text{3b}}{}$
}\end{gathered}\right)+P_n\left(\begin{gathered}\text{\tiny \usefont{U}{dice3d}{m}{n}$\stackrel{\text{3b}}{}$ 
 $\stackrel{\text{3b}}{}$}\end{gathered}\right)$
 using $n$ indistinguishable bosons.  Let $\rho^{(n)}$ denote the density matrix of $n$ 6-level bosons, and $\Pi_{sym}^{(n)}$ be the symmetriser for the $n$ bosons. {Note that, we consider 6-level bosons because we are focusing on a dice rolling experiment. Here $n$ corresponds to the number of rolls ($n=r$).} 
{ From Theorems \ref{th:repr41}--\ref{th:2}, we can rewrite $v_n$ as 
 $v_n=\min_{\rho^{(n)}} Tr(D^{[n]}\rho^{(n)})$ with $\rho^{(n)}=\Pi_{sym}^{(n)} \rho^{(n)} \Pi_{sym}^{(n)}$, $Tr(\rho^{(n)})=1$ and $\rho^{(n)} \succeq 0$. 
All the superscripts are just indexes. Note that,  the superscript in $D^{[n]}$ denotes the number of bosons (the number of rolls in the dice experiment). For instance, for $n=2$, 
$D^{[2]}$ can be written as an algebraic sum of the matrices $D^{(t_1t_2)}$. In particular, we have that
$D^{[2]}=D^{(t_1=1,t_2=1)}-0.5D^{(t_1=1,t_2=2)}-0.5D^{(t_1=2,t_2=1)}+D^{(t_2=2,t_2=2)}$ and it is a diagonal matrix with all zero entries except $D^{[2]}_{1}=1=D^{[2]}_{8}$ and $D^{[2]}_{2}=D^{[2]}_{7}=-0.5$ ($D^{[2]}$ is a physical-observable operator, as it is symmetric to permutations). } This holds because the diagonal of $\rho^{(2)}$ corresponds to the monomials
$[\theta_1^2,\theta_1\theta_2,\dots,\theta_2\theta_1,\theta_2^2,\dots,\theta_6^2]$. Any $D^{[n]}$ (for $n>2$) can be defined in a similar way exploiting the  `partial trace'-style operations illustrated in Eq.~\eqref{eq:partialtrex}.

The value of $v_n$ is shown in Figure \ref{fig:convQP} as a function of $n$. Note that, in this case, in order to compute  $v_n=\min_{\rho^{(n)}} Tr(D^{[n]}\rho^{(n)})$, we would in principle need to solve a semidefinite programming problem due to the constraint $\rho^{(n)}\succeq0$ (it must be a PSD matrix). However, since $D^{[n]}$ is diagonal, this problem reduces to a linear programming problem.

\begin{figure}[h]
    \centering
    \includegraphics[width=7cm]{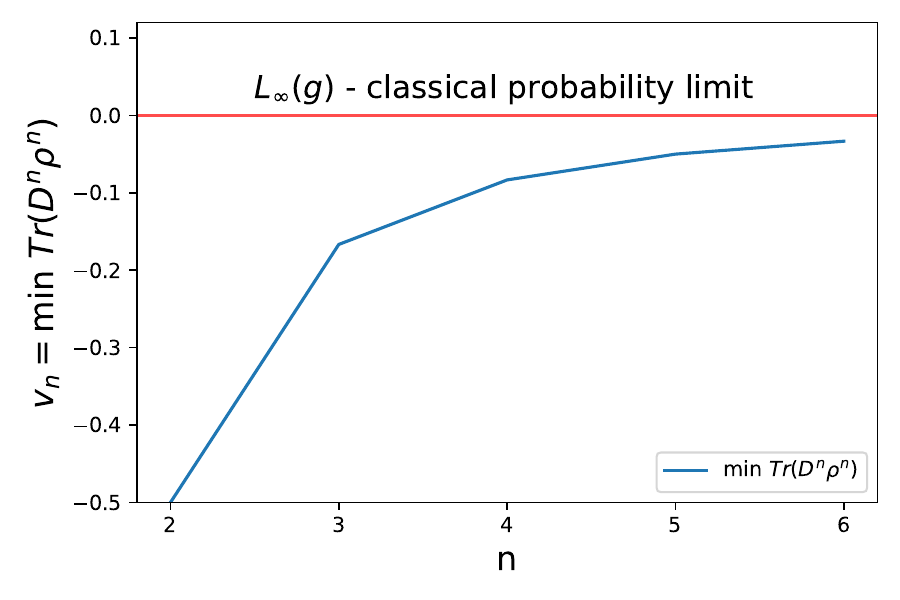}
    \caption{Value of $v_n$ as a function of $n$. The red line corresponds to $v_{\infty}$.}
    \label{fig:convQP}
\end{figure}

The values of $v_n$ in Figure \ref{fig:convQP} coincide exactly with the values in Figure \ref{fig:convLP}.

This means that the worst-case marginal expectation for an operator $D^{[n]}$ on two bosons, which are part of a longer sequence of additional 0, 1, 2,\dots bosons, converges to the
classical probability value as the total number of bosons increases, see Figure \ref{fig:dicebosons}. In QT, we would refer to this phenomenon as the `emergence of classical reality'. In probability theory, this is the convergence from finite to infinite exchangeability.\\

\begin{figure}[h]
    \centering
    \includegraphics[width=8cm]{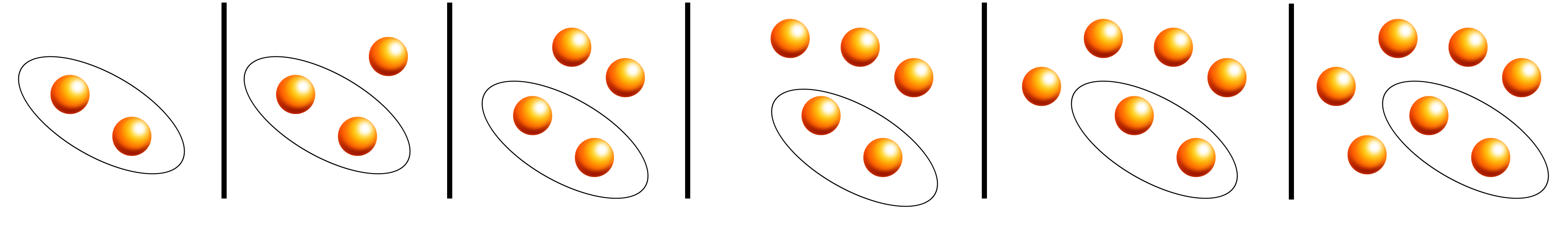}
    \caption{The worst-case marginal expectation for two bosons, which are part of a longer  indistinguishable sequence of additional $0,1,2,\dots$ bosons, converges to the classical 
    probability value.}
    \label{fig:dicebosons}
\end{figure}
\end{example}
The example above illustrates the following connection between finite exchangeability and entanglement. To understand it, for simplicity, we consider the case when $r=d=2$, which corresponds to a coin tossing experiment with two coins. 
By  Theorem \ref{th:2}, the finitely exchangeable sequence $t_1,t_2$ with respect to the probability $P_2(t_1=H,t_2=T)=P_2(t_1=T,t_2=H)=0.5$ (H=Heads, T=Tails) can be represented   as
\begin{equation}
\label{eq:rhosymm}
\begin{array}{l}
\rho=\begin{bmatrix}
0 & 0 & 0 & 0\\
0 & 0.5 & 0.5 & 0\\
0 & 0.5 & 0.5 & 0\\
0 & 0 & 0 & 0\\
\end{bmatrix}\vspace{2mm}\\
=M=\widehat{L}
\resizebox{2.5in}{!}{$\left( \begin{matrix}
(x_1^2)^\dagger x_1^2& (x_1^2)^\dagger x_1 x_2 & (x_1^2)^\dagger x_1 x_2 & (x_1^2)^\dagger x_2 x_2\\
(x_1x_2)^\dagger x_1^2& (x_1x_2)^\dagger x_1 x_2 & (x_1x_2)^\dagger x_1 x_2 & (x_1x_2)^\dagger x_2 x_2\\
(x_1x_2)^\dagger x_1^2& (x_1x_2)^\dagger x_1 x_2 & (x_1x_2)^\dagger x_1 x_2 & (x_1x_2)^\dagger x_2 x_2\\
(x_2^2)^\dagger x_1^2& (x_2^2)^\dagger x_1 x_2 & (x_2^2)^\dagger x_1 x_2 & (x_2^2)^\dagger x_2 x_2\\
\end{matrix}\right)$}
\end{array}   
\end{equation}
where $\widehat{L}((x_1x_2)^\dagger x_1 x_2)=0.5$ and zero otherwise.

The diagonal elements of $\rho$ include the probabilities $\rho_{11}=P_2(t_1=H,t_2=H)=0$, $\rho_{22}=P_2(t_1=H,t_2=T)=0.5$,  $\rho_{33}=P_2(t_1=T,t_2=H)=0.5$  and $\rho_{44}=P_2(t_1=T,t_2=T)=0$.
The symmetry condition $\rho = \Pi_{sym}\rho \Pi_{sym}$ implies that the off-diagonal elements $\rho_{23},\rho_{32}$ are equal to $0.5$. 
In fact, when considering the matrix defined by $\widehat{L}$, the four elements $M_{22},M_{23},M_{32},M_{33}$ are equal: by definition, the operator $\widehat{L}$ applies to the same monomial $x_1^\dagger x_2^\dagger x_1 x_2$.

In Example \ref{ex:last}, we then chose a Hermitian operator $D^{[2]}$, which for the two coins case reduces to  $D^{[2]}=\text{diag}([1,-0.5,-0.5,1])$, and results in the value $Tr(D^{[2]}\rho)=-0.5$. This value is the same value obtained by solving  $v_2=\min_{P_2} P_2(H,H)-P_2(H,T)+P_2(T,T)=0-0.5+0=-0.5$.

In Example \ref{ex:thetax}, we  observed that the coefficients in the diagonal of $D^{[2]}$ give rise to the polynomial $\theta_1^2-\theta_1\theta_2+\theta_2^2$. This polynomial is always nonnegative and, therefore, can serve as a \textit{finite exchangeability witness}.
In fact, we used the inequality $-0.5=v_2<v_{\infty}=0$ to prove that $P_2$ is finitely exchangeable.
The reader may thus ask: is $D^{[2]}$ also an entanglement witness for $\rho$, as defined in Eq.~\eqref{eq:rhosymm}?\footnote{{In QT, in the case of distinguishable particles,  witnesses are commonly \citep{horodecki2009quantum} used to differentiate entangled density matrices from separable (non-entangled) ones, which are compatible with classical probability theory.}}
Several nonequivalent ways of identification and quantification of the entanglement for bosons have been  proposed, see \cite{benatti2020entanglement} for a recent review. 
One of the main notions of entanglement for bosons states that non-entangled states are equivalent to \textit{simple symmetric tensors} \cite{pavskauskas2001quantum,eckert2002quantum,wang2005canonical,zhou2009correlation,kus2009classical,kotowski2010universal,grabowski2012segre,grabowski2011entanglement,sawicki2012critical,hyllus2012entanglement,oszmaniec2013universal,wasak2014cauchy,oszmaniec2014fraction,sawicki2014convexity,rigolin2016entanglement,garcia2017buildup,karczewski2019sculpting,morris2020entanglement,kunkel2018spatially}. Following \cite{benatti2020entanglement}, we refer to this definition of entanglement for identical particles as \textit{Entanglement-I}. In the case of two bosons, this definition implies that a pure state is non-entangled if and only
if it is a tensor product of two identical one-particle states \cite{grabowski2011entanglement}.
Notice that this definition is consistent with the model introduced through the operator $ \widehat{L}$ in Theorem \ref{th:repr41}. { Indeed, a pure\footnote{{A pure (or atomic) probability measure, corresponding to a non-entangled case, is one that assigns a precise probability distribution concentrated on a single outcome ${\bf x}$. In other words, it represents complete certainty about ${\bf x}$, as the probability measure cannot be decomposed into a convex combination of other measures.
}} non-entangled density matrix is simply the matrix $M=(\otimes_r \hat{\bf x})(\otimes_r \hat{\bf x})^\dagger$ for some $ \hat{{\bf x}} \in \mathbb{C}^d$ with $\hat{{\bf x}}^\dagger \hat{{\bf x}}=1$, corresponding to the product of identical one-particle states (each state corresponding to $\hat{{\bf x}}$).  In the dice experiment, this means that we have complete knowledge of the probabilities of each face of the dice.} Therefore,  $D^{[2]}=\text{diag}([1,-0.5,-0.5,1])$ is a witness for this definition of entanglement. Indeed, observe that the polynomial $(\otimes_2 {\bf x})^\dagger D^{[2]} (\otimes_2 {\bf x})$ is nonnegative. Since $Tr(D^{[2]}\rho)=-0.5<0$, this means that $\rho$ in Eq.~\eqref{eq:rhosymm} is entangled. 
{ 
\begin{remark}
To clarify this point, consider the polynomial 
$$
(\otimes_2 {\bf x})^\dagger D^{[2]} (\otimes_2 {\bf x})=\theta_1^2-\theta_1\theta_2+\theta_2^2.
$$
This polynomial is nonnegative for every probabilities $[\theta_1,\theta_2]$ we choose for the faces of the coin.
According to the \textit{Entanglement-I} definition, the only pure non-entangled states are those where $[\theta_1,\theta_2]$ are fully known. In other words, in this pure case,  an operator $\widehat{L}$ is not entangled if $\widehat{L}\left((\otimes_2 {\bf x})(\otimes_r {\bf x})^\dagger\right)=(\otimes_2 \hat{{\bf x}})(\otimes_2 \hat{\bf x})^\dagger$ for some $ \hat{{\bf x}} \in \mathbb{C}^2$ with $\hat{{\bf x}}^\dagger \hat{{\bf x}}=1$. A pure state $\hat{\mathbf{x}}$ attains the infimum of $(\otimes_2 {\bf x})^\dagger D^{[2]} (\otimes_2 {\bf x})$, since any mixture of pure states yields a higher value. 
Therefore, since (i) $(\otimes_2 {\bf x})^\dagger D^{[2]} (\otimes_2 {\bf x})\geq 0$; (ii) the operator $\widehat{L}$ defined as $\widehat{L}\left((\otimes_2 {\bf x})(\otimes_2 {\bf x})^\dagger\right)=\rho$, with $\rho$ in Eq.~\eqref{eq:rhosymm}, is such that 
$\widehat{L}\left((\otimes_2 {\bf x})^\dagger D^{[2]} (\otimes_2 {\bf x})\right)<0$; this means that this operator $\widehat{L}$ must be \textit{Entangled-I} (it cannot be written as $\rho=\sum_{i}w_i(\otimes_2 \hat{{\bf x}}_i)(\otimes_2 \hat{\bf x}_i)^\dagger$ for any probabilistic mixture of  $\hat{{\bf x}}_i$).
\end{remark}}
Entanglement-I has been challenged as a definition of entanglement for indistinguishable particles \cite{benatti2020entanglement}. We will discuss these arguments in the next Section \ref{sec:application}.

\section{Application}
\label{sec:application}
 For indistinguishable particles (IPs) systems, many (of the standard equivalent) notions of entanglement are ambiguous.
The reason is that  the symmetries resulting from  indistinguishability have mathematically the same structure as a conventional superposition. Therefore, symmetries generate a sort of \textit{structural entanglement}, which is simply determined by the particles fundamental indistinguishability. { For instance, the density matrices in~\eqref{eq:exrho} are Bose-symmetric, and therefore valid density matrices for indistinguishable bosons. Their structure is constrained by this symmetry. However, in the case of distinguishable particles, both matrices correspond to density matrices of two entangled particles. This creates an ambiguity between the symmetries arising from indistinguishability and the structural features that give rise to entanglement in the distinguishable case.
 }
This has given rise to different viewpoints about physical meaning and assessment of entanglement for IPs. There at least five main different definitions of entanglement for IPs in literature -- we refer the reader to \cite{benatti2020entanglement} for an overview, and a survey of the literature on the argument.

 Recently, \cite{benatti2020entanglement} attempted to systematise these different definitions by proposing a set of properties that should be used to characterise entanglement for IPs. These properties have been defined by trying to generalise properties that hold for entanglement of distinguishable particles.

The proposed properties are: (i) \textit{local operators:} entanglement for IPs must correspond to the presence of non-local correlations; (ii) \textit{effective distinguishability:} the entanglement theory developed for
IPs must reduce to the one for distinguishable particles when applied to IPs that have effectively become distinguishable by ``freezing''
suitable degrees of freedom; (iii) \textit{information processing resources:}  local operators acting on separable states must not enhance the
performances of informational tasks with respect to classical ones.
It is important to note that \cite{benatti2020entanglement} does not relate  entanglement to the violation of Bell-type inequalities ({ that is, through Dutch-book arguments}).  

In \cite{benatti2020entanglement}, it is claimed that `Entanglement-I' violates the properties \textit{local operators} and \textit{effective distinguishability}. We  will use the analogy between finite exchangeability and quantum theory developed in this paper to highlight some issues in the definitions of \textit{local operators} used in \cite{benatti2020entanglement}, and explain the argument about \textit{effective distinguishability}. In particular, we will show that the definitions given in \cite{benatti2020entanglement} lead to contradictions when applied (through our connection between QT for bosons and classical finite exchangeability) to classical finitely exchangeable probabilities. We then propose a modification of these properties, which avoids this issue.

\subsection{Local operators}
\label{sec:locop}
We start first with the argument about \textit{local operators} \cite[Sec.\ 3.1.1]{benatti2020entanglement}. 
Consider two single-particle operators  ({ Hermitian matrices}) $O_1$ for particle $1$ and $O_2$ for particle $2$. We know that
$$
O_1 \otimes O_2 = (O_1 \otimes I_2)(I_1 \otimes O_2), 
$$
where $I_i$ is the identity-matrix for the i-th particle. For distinguishable particles, it is well known that a bipartite density
matrix $\rho_{12}$ is separable if and only if, for all $O_1,O_2$:
\begin{equation}
Tr((O_1 \otimes O_2)\,\rho_{12})=\sum_{j=1}^m p_j Tr(O_1  \rho_j^{(1)}) Tr(O_2  \rho_j^{(2)}),
\end{equation}
for some probabilities $p_j$, $\sum_{j=1}^m p_j=1$, and single-particle density matrices $\rho_j^{(1)}$ and, respectively $\rho_j^{(2)}$. For a pure density matrix $\rho_{12}$, the above equation implies that:
\begin{equation}
\label{eq:localoperclass}
\begin{aligned}
  &Tr((O_1 \otimes O_2) \rho_{12})-Tr((O_1 \otimes I_2)\, \rho_{12})Tr((I_1 \otimes O_2)\, \rho_{12})=0. 
\end{aligned}
\end{equation}
The extension of this definition to IPs is quite challenging because it requires that each particle is individually
addressable. For IPs, $O_1 \otimes O_2$, $O_1 \otimes I_2$ and $I_1 \otimes O_2$ are not physical observable in general, because they are not symmetric.

To account for this issue,   \cite[Sec.\ 3.1.1]{benatti2020entanglement} considers  two commuting operators $O_1,O_2$ and define their symmetrised versions $\mathcal{P}(O_1,I),\mathcal{P}(I,O_2)$, as follows:
\begin{align}
\label{eq:symmOper}
\mathcal{P}(O,I)=\mathcal{P}(I,O)=O \otimes I + I \otimes O.
\end{align}
Commutativity of $O_1,O_2$ implies that  they  have a common eigenbasis $\{\ket{e_\lambda}\}_\lambda$ corresponding to eigenvalues $\{o^{(1)}_\lambda\}_\lambda$ for $O_1$, and $\{o^{(2)}_\lambda\}_\lambda$ for $O_2$. Now consider a
non-entangled-I state of the form $\ket{\psi}\otimes\ket{\psi}$ with
\begin{equation}
\ket{\psi}=\frac{\ket{e_\lambda}+\ket{e_\kappa}}{\sqrt{2}},
\end{equation}
and $\lambda\neq\kappa$, { where $\ket{\psi}$ denotes a column-vector in $ \mathbb{C}^2$,  $\bra{\psi}$ denotes its transpose; each vector $\ket{\psi}$  has unit-norm, that is $\braket{\psi}=1$.}

Then, the analogous (according to the definition in \cite{benatti2020entanglement}) of \eqref{eq:localoperclass} for a non-entangled-I state becomes:
\begin{align}
& \langle\psi\otimes\psi|\mathcal{P}(O_1,I)\mathcal{P}(O_2,I)\ket{\psi\otimes\psi}- \nonumber \\
& \langle\psi\otimes\psi|\mathcal{P}(O_1,I)\ket{\psi\otimes\psi}\langle\psi\otimes\psi|\mathcal{P}(O_2,I)\ket{\psi\otimes\psi} \nonumber \\
& =\frac{1}{2}\big(o^{(1)}_\lambda-o^{(1)}_\kappa\big)\big(o^{(2)}_\lambda-o^{(2)}_\kappa\big)\ .
\label{fact.id.fact.ent}
\end{align}
The above difference is zero if and only if either $o^{(1)}_\lambda=o^{(1)}_\kappa$ or $o^{(2)}_\lambda=o^{(2)}_\kappa$, which is not true in general for generic commuting $O_1,O_2$. According to the definition of \textit{local operators} in \cite{benatti2020entanglement}, this means that, in general, non-entangled-I states do not satisfy the factorisation of expectations for (symmetric) local operators, which ({ in the distinguishable particles case}) highlights the presence of non-local correlation  ({ entanglement}).

\subsubsection{Analogy with  finite/infinite  exchangeability}
\label{sec:an1}
Hereafter, by using the analogy between finite exchangeability and QT, we challenge this conclusion and, in particular, the definition \eqref{eq:symmOper} of local operators. 

Consider a coin tossing experiment with two coins ($r=d=2$). { We choose the following diagonal matrices to define polynomials over the probabilities of the coin sides.}
\begin{equation}
\label{eq:O1O2_1}
O_1 =\begin{bmatrix}
1 & 0\\
0 & 0
\end{bmatrix}, ~~O_2 =\begin{bmatrix}
0 & 0\\
0 & 1
\end{bmatrix}.
\end{equation}
{ They commute in the sense that $O_1O_2-O_2O_1=0$. These diagonal matrices are defined for the single coins, we then extend them to two coins using the approach described in \cite[Sec.\ 3.1.1]{benatti2020entanglement}.}
This results in the following symmetrised versions:
$$
\begin{array}{rcl}
\mathcal{P}(O_1,I) &=&\begin{bmatrix}
2 & 0 & 0 & 0\\
0 & 1 & 0 & 0\\
0 & 0 & 1 & 0\\
0 & 0 & 0 & 0\\
\end{bmatrix},\vspace{1mm}\\
\mathcal{P}(I,O_2) &=&\begin{bmatrix}
0 & 0 & 0 & 0\\
0 & 1 & 0 & 0\\
0 & 0 & 1 & 0\\
0 & 0 & 0 & 2\\
\end{bmatrix},\vspace{1mm}\\
\mathcal{P}(O_1,I)\mathcal{P}(I,O_2)  &=&\begin{bmatrix}
0 & 0 & 0 & 0\\
0 & 1 & 0 & 0\\
0 & 0 & 1 & 0\\
0 & 0 & 0 & 0\\
\end{bmatrix},
\end{array}
$$
which are all diagonal operators and, therefore, they are valid observables for the  tossing of two finitely exchangeable coins. Indeed, these operators define the following polynomials
\begin{align}
\nonumber
g_1(\x^\dagger,\x) &=(\otimes_2 {\bf x})^\dagger \mathcal{P}(O_1,I) (\otimes_2 {\bf x})\\
&=2(x_1^\dagger x_1)^2+2(x_1^\dagger x_1)(x_2^\dagger x_2) \\
&=2\theta_1^2+2\theta_1\theta_2 \\
\nonumber
g_2(\x^\dagger,\x) &=(\otimes_2 {\bf x})^\dagger \mathcal{P}(I,O_2) (\otimes_2 {\bf x})\\
&=2(x_1^\dagger x_1)(x_2^\dagger x_2)+2(x_2^\dagger x_2)^2\\
&=2\theta_1\theta_2+2\theta_2^2 \\
\nonumber
g_{12}(\x^\dagger,\x) &=(\otimes_2 {\bf x})^\dagger \mathcal{P}(O_1,I)\mathcal{P}(I,O_2)(\otimes_2 {\bf x})\\
&=2(x_1^\dagger x_1)(x_2^\dagger x_2)\\
&=2\theta_1\theta_2.
\end{align}

For a pure state (that is, for a choice of the probability of Heads $\theta_1$ and Tails $\theta_2$ of the coin), the \textit{local operators} condition \eqref{fact.id.fact.ent} requires that:
$$
\begin{aligned}
&2\theta_1\theta_2-(2\theta_1^2+2\theta_1\theta_2)(2\theta_1\theta_2+2\theta_2^2)\\
&=-2 \theta_1(1-\theta_1)=0,\\
\end{aligned}
$$
(we have exploited that $\theta_2=1-\theta_1$), which is zero if and only if either $\theta_1=0$ or $\theta_1=1$. The \textit{local operators} condition \eqref{fact.id.fact.ent} is therefore violated by the independent tossing of two identical coins with $\theta_1\in (0,1)$.  For instance, for  $\theta_1=0.5$, the difference is equal to $-2 \theta_1(1-\theta_1)=-0.5$, which coincides with the value computed by using  \eqref{fact.id.fact.ent} { (for $O_1,O_2$ defined in \eqref{eq:O1O2_1})}.
This  shows that the argument about \textit{local operators} is problematic. If it were true -- { by exploiting the connection between entangled states being equivalent to finite exchangeability and non-entangled states being equivalent to infinite exchangeability} --  this would imply that the uniform probability  $P(t_1=H,t_2=H)=\theta_1^2=0.25$, $P(t_1=H,t_2=T)=P(t_1=T,t_2=H)=\theta_1\theta_2=0.25$ and $P(t_1=T,t_2=T)=(1-\theta_1)(1-\theta_2)=0.25$ is finitely (entangled) but not infinitely exchangeable  (not entangled), which is clearly false.

The key issue here is the definition of local operators, and we will use the analogy with exchangeability to properly define them. From the results in Section \ref{sec:Ent}, it is immediate verify that, for any  an observable $G$,
    $$
    \begin{aligned}\widehat{L}\Big( (\otimes_r {\bf x})^\dagger G (\otimes_r {\bf x})\Big)
    &=\widehat{L}\Big( (\otimes_r {\bf x})^\dagger \tfrac{\Pi_i^\dagger G\Pi_j +\Pi_j^\dagger G\Pi_i}{2} (\otimes_r {\bf x})\Big),
    \end{aligned}$$ for all $\Pi_i,\Pi_j \in \mathbb{P}_2$, where $\mathbb{P}_2$ denotes the  collection of all permutation matrices of two particles,
    and
    $$\begin{aligned}\widehat{L}\Big( (\otimes_r {\bf x})^\dagger G (\otimes_r {\bf x})\Big)
    &=\widehat{L}\Big( (\otimes_r {\bf x})^\dagger \Pi_{sym}G \Pi_{sym} (\otimes_r {\bf x})\Big).        
    \end{aligned}$$
{ This holds because $ \Pi_{sym} (\otimes_r {\bf x})= (\otimes_r {\bf x})$ from Theorem \ref{th:2}.}
The symmetries in exchangeability (and, analogously, indistinguishability) thus define \textit{equivalence classes} of observables,  consisting of all permuted gambles with the same expectation as $Tr(G M)$, where $M=\widehat{L}\Big( (\otimes_r {\bf x})(\otimes_r {\bf x})^\dagger  \Big)$.
This equivalence class contains physical observables, that is operators that are  symmetric to permutations. We can thence calculate the expectation value of a non-physical observable $G$ by simply choosing any physical observable in the equivalence class $\{G\}^{eq}$ and computing its expectation value.
It is immediate to show 
that
\begin{align}
&\frac{O_1 \otimes I+I \otimes O_1}{2} \in \{O_1 \otimes I\}^{eq}\\
&\frac{I \otimes O_2+O_2 \otimes I}{2} \in \{I \otimes O_2\}^{eq}\\
&\frac{O_1 \otimes O_2+O_2 \otimes O_1}{2} \in \{O_1 \otimes O_2\}^{eq}
\end{align}
In fact, notice for instance that
$$\Pi_{sym} \frac{O_1 \otimes O_2+O_2\otimes O_1}{2} \Pi_{sym}=\Pi_{sym} (O_1 \otimes  O_2) \Pi_{sym}$$
for each $O_1,O_2$.
 Therefore, it holds that
{
 \begin{equation}
\begin{aligned}
& (\otimes_2 {\bf x})^\dagger\tfrac{O_1 \otimes O_2+O_2 \otimes O_1}{2}(\otimes_2 {\bf x})-  \\
& (\otimes_2 {\bf x})^\dagger\tfrac{O_1 \otimes I+I \otimes O_1}{2}(\otimes_2 {\bf x})\\
&\cdot (\otimes_2 {\bf x})^\dagger\tfrac{I \otimes O_2+O_2 \otimes I}{2}(\otimes_2 {\bf x}) \\
&=(\otimes_2 {\bf x})^\dagger(O_1 \otimes O_2)(\otimes_2 {\bf x})- \nonumber \\
&(\otimes_2 {\bf x})^\dagger(O_1 \otimes I)(\otimes_2 {\bf x})\\
&\cdot (\otimes_2 {\bf x})^\dagger(I \otimes O_2)(\otimes_2 {\bf x})
=0,
\end{aligned}
\label{fact.id.fact.entA}
\end{equation}
or, using the bra-ket notation:}
 \begin{equation}
\begin{aligned}
& \bra{\psi\otimes\psi}\tfrac{O_1 \otimes O_2+O_2 \otimes O_1}{2}\ket{\psi\otimes\psi}-  \\
& \bra{\psi\otimes\psi}\tfrac{O_1 \otimes I+I \otimes O_1}{2}\ket{\psi\otimes\psi}\\
&\cdot \bra{\psi\otimes\psi}\tfrac{I \otimes O_2+O_2 \otimes I}{2}\ket{\psi\otimes\psi} \\
&=\bra{\psi\otimes\psi}(O_1 \otimes O_2)\ket{\psi\otimes\psi}- \nonumber \\
& \bra{\psi\otimes\psi}(O_1 \otimes I)\ket{\psi\otimes\psi}\\
&\cdot \bra{\psi\otimes\psi}(I \otimes O_2)\ket{\psi\otimes\psi} 
=0.
\end{aligned}
\label{fact.id.fact.entA}
\end{equation}
This shows that the difference in \eqref{fact.id.fact.entA} is zero when we account for the constant $1/2$ in the computations. To sum up, entangled-I states satisfy the factorisation of expectations for (symmetric) local operators, once we define these operators within the equivalence class of the local operators for distinguishable particles.

\subsection{Effective distinguishability}
As done in \cite{benatti2020entanglement}, we  consider now 
 particles described by an ``external'' degree of freedom with values
$L,R$, and an ``internal'' degree of freedom with values $\uparrow,\downarrow$. Therefore, the single-particle Hilbert space is  spanned by the  orthogonal states $\ket{L,\uparrow},\ket{L,\downarrow},\ket{R,\uparrow},\ket{R,\downarrow}$. Among the vectors of the two-particle symmetrized Hilbert
space, we focus on  those of the form \cite{benatti2020entanglement}:
\begin{align}
\label{eq:indi1}
\ket{\phi^\textnormal{ind}_1}
  & =\frac{1}{\sqrt{2}}\Big[  \ket{ L,\uparrow}\otimes\ket{ R,\downarrow} +
  \ket{ R,\downarrow}\otimes\ket{ L,\uparrow} \Big].\ 
\end{align}
Suppose $L$ and $R$ are orthogonal single-particle states, such as when particles are in separate spatial regions (left and right), so that $\langle L\ket{ R}=0$. Then the external degree of freedom, which can be identified as position, can be used to distinguish the particles, as if they were labelled. Indeed, by measuring and then fixing the spatial degree of freedom, one can assign the particle in the region $L$ as particle $1$, and the particle in $R$, as particle $2$. 

The label attribution is incompatible with particle identity. However,
with respect to the observables $\mathcal{P}(P_L\otimes O_1,P_R\otimes O_2)$, where  
$P_{L,R}$ are the projectors $\ket{ L}\bra{ L}$ and $\ket{ R}\bra{R}$ onto states localised in the left, respectively right region, while $O_1$ and $O_2$ are arbitrary observables relative to the internal degree of freedom,  we cannot distinguish the states satisfying Eq. \eqref{eq:indi1}  from those of the form
\begin{align}
\label{eq:dis1}
\ket{\phi^\textnormal{dist}_1} & =\ket{L,\uparrow}\otimes\ket{R,\downarrow}.
\end{align}
In other words, 
the mean values
of the  observables $\mathcal{P}(P_L\otimes O_1,P_R\otimes O_2)$ satisfy the equalities:
\begin{align}
& \bra{\phi^\textnormal{ind}_1}\mathcal{P}\big(P_L\otimes O_1,P_R\otimes O_2\big)\ket{\phi^\textnormal{ind}_1}= \nonumber \\
& =\bra{L,\uparrow}P_L\otimes O_1\ket{L,\uparrow}\bra{R,\downarrow}P_R\otimes O_2\ket{R,\downarrow}\nonumber\\
\label{eq:distind}
& =\bra{\phi^\textnormal{dist}_1} \big(P_L\otimes O_1\big)\otimes \big(P_R\otimes O_2\big)\ket{\phi^\textnormal{dist}_1}\ .
\end{align}

{
The same equality is also true for 
\begin{align}
\label{eq:dis2}
\ket{\phi^\textnormal{dist}_2} & =\ket{R,\downarrow} \otimes \ket{L,\uparrow},
\end{align}
as
\begin{align}
& \bra{\phi^\textnormal{ind}_1}\mathcal{P}\big(P_L\otimes O_1,P_R\otimes O_2\big)\ket{\phi^\textnormal{ind}_1}= \nonumber \\
\label{eq:distind1}
& =\bra{\phi^\textnormal{dist}_2}  \big(P_R\otimes O_2\big)\otimes \big(P_L\otimes O_1\big) \ket{\phi^\textnormal{dist}_2},
\end{align}
where  $\ket{\phi^\textnormal{dist}_1},\ket{\phi^\textnormal{dist}_2}$ are the ``non-symmetric'' parts of $\ket{\phi^\textnormal{ind}_1}$.
}

According to \cite{benatti2014entanglement}, this is an argument against Entanglement-I. The reason is that one can construct two-particle states of the form \eqref{eq:indi1}, therefore
Entangled-I, that, by using spatial orthogonality, can be made effectively correspond to distinguishable particle separable states, such as those satisfying Eq. \eqref{eq:dis1} (or Eq. \eqref{eq:distind1}).
Therefore, \cite{benatti2014entanglement} claims that distinguishable particle-entanglement cannot be recovered from Entanglement-I.

\subsubsection{Analogy with  finite/infinite  exchangeability}
Hereafter, we leverage again the similarities between finite exchangeability and indistinguishability to clarify this counter-example. The derivations presented in Section \ref{sec:Ent} need a modification to account for the inherent  partial exchangeability  of this example.

We make the analogy using a pair of coins, and we assume that the possibility space is given by:
$$
\left\{\coin{L\strut}\coin{$\uparrow$\strut},\coin{L\strut}\coin{$\downarrow$\strut},\coin{R\strut}\coin{$\uparrow$\strut},\coin{R\strut}\coin{$\downarrow$\strut}\right\}
$$
{ We have two types of coins, one with faces $L/R$ and the other with faces  $\uparrow/\downarrow$.}
In this case, by
$t_1, t_2,\dots,t_r$, we denote the results of $r$-tosses of the pair, such as for example:
$$
\underset{t_1}{\coin{R\strut}\coin{$\uparrow$\strut}}~~~\underset{t_2}{\coin{L\strut}\coin{$\uparrow$\strut}}~~~\underset{t_3}{\coin{L\strut}\coin{$\uparrow$\strut}}.
$$

We assume that the sequence $t_1, t_2,\dots,t_r$ is finitely exchangeable. For two tosses of the pair, this for instance means that
$$
\begin{aligned}
P\left(\coin{R\strut}\coin{$\uparrow$\strut},\,\coin{L\strut}\coin{$\uparrow$\strut}\right)=P\left(\coin{L\strut}\coin{$\uparrow$\strut},\,\coin{R\strut}\coin{$\uparrow$\strut}\right),\\
P\left(\coin{R\strut}\coin{$\uparrow$\strut},\,\coin{L\strut}\coin{$\downarrow$\strut}\right)=P\left(\coin{L\strut}\coin{$\downarrow$\strut},\,\coin{R\strut}\coin{$\uparrow$\strut}\right).\\
\end{aligned}
$$
Technically, this is a definition of \textit{finite partial exchangeability} as proposed by de Finetti in \cite{de1980condition}. Noticed that the definition assumes the pair being exchangeable, not the single coins. This means that in general:
\begin{align}
\label{eq:noexchcase}
P\left(\coin{R\strut}\coin{$\uparrow$\strut},\,\coin{L\strut}\coin{$\downarrow$\strut}\right)\neq P\left(\coin{L\strut}\coin{$\uparrow$\strut},\,\coin{R\strut}\coin{$\downarrow$\strut}\right),
\end{align}
where we have exchanged only the first coin in the pair. 
We then adapt Proposition \ref{prop:kerns} as follows. Let $\theta_L,\theta_R$ be the probability for face-L and face-R and $\theta_\uparrow,\theta_\downarrow$ the probability for face-$\uparrow$ and face-$\downarrow$.  Then we have that:
 \begin{equation}
\label{eq:reprThsigned0Coins}
\begin{aligned}
  &P(t_1,\dots,t_r) \\
  &= \int_\Theta (\theta_L\theta_{\uparrow})^{n_1}(\theta_L\theta_{\downarrow})^{n_2} 
(\theta_R\theta_{\uparrow})^{n_3}(\theta_R\theta_{\downarrow})^{n_4}
d\nu({\boldsymbol \theta}),  
\end{aligned}
\end{equation}
where $n_1$ is the number of times the pair landed on $L,\uparrow$,  $n_2$ is the number of times the pair landed on $L,\downarrow$ and so on.

{
To make the connection with QT in the context under consideration, we introduce a  complex vector ${\bf x}=[x_{11},x_{12},x_{21},x_{22}]^\top$ in $\mathbb{C}^4$, such that ${\bf x}^\dagger {\bf x}=1$. We can then express the probabilities of the four outcomes of the two coins as:
$$
\begin{aligned}
&[\theta_L\theta_{\uparrow},\theta_L\theta_{\downarrow},\theta_R\theta_{\uparrow},\theta_R\theta_{\downarrow}]=[x_{11}^\dagger x_{11}, x_{12}^\dagger x_{12},
x_{21}^\dagger x_{21}, x_{22}^\dagger x_{22}].
\end{aligned}
$$
and consider $r=2$ rolls.

Observe that, expressed in the ``language of coins'', the assumption that the particles are in separate spatial regions (left and right) means that  
$\widehat{L}(\theta_L\theta_\uparrow\theta_L\theta_\uparrow)=\widehat{L}((x_{11}^\dagger x_{11})^2)=0$, and also $\widehat{L}((x_{11}^\dagger x_{11})(x_{12}^\dagger x_{12})),\widehat{L}((x_{12}^\dagger x_{12})^2)$ are all zeros (similarly for the quasi-expectations obtained replacing $\theta_L$ by $\theta_R$).   In other words, we assign zero quasi-expectation to  the cases in which the coin with face $L/R$ lands on the same side, either $L$ or $R$. 
There remain only eight cases:
{\small
$$
\begin{aligned}
&\left(\coin{L\strut}\coin{$\uparrow$\strut},\,\coin{R\strut}\coin{$\uparrow$\strut}\right), \left(\coin{L\strut}\coin{$\downarrow$\strut},\,\coin{R\strut}\coin{$\downarrow$\strut}\right),\left(\coin{L\strut}\coin{$\uparrow$\strut},\,\coin{R\strut}\coin{$\downarrow$\strut}\right),\left(\coin{L\strut}\coin{$\downarrow$\strut},\,\coin{R\strut}\coin{$\uparrow$\strut}\right),\\
&\left(\coin{R\strut}\coin{$\uparrow$\strut},\,\coin{L\strut}\coin{$\uparrow$\strut}\right), \left(\coin{R\strut}\coin{$\downarrow$\strut},\,\coin{L\strut}\coin{$\downarrow$\strut}\right),\left(\coin{R\strut}\coin{$\downarrow$\strut},\,\coin{L\strut}\coin{$\uparrow$\strut}\right),\left(\coin{R\strut}\coin{$\uparrow$\strut},\,\coin{L\strut}\coin{$\downarrow$\strut}\right).\\
\end{aligned}
$$}
The elements within the same column are exchangeable. By leveraging the spatial degrees of freedom to distinguish the two pairs of coins, one effectively uses the $L/R$ outcomes as contextual variables to label the coin showing the $\uparrow$ or $\downarrow$ face as ``left'' or ``right''. In other words, the outcome $L/R$ is used to select one of the rows above. 
Since the probabilities associated with elements in each row are not finitely exchangeable (as noted in \eqref{eq:noexchcase}), they are not subject to any exchangeability constraints.  

Therefore, the argument in \cite{benatti2020entanglement} against ``Entanglement I'' (that is, that one can construct two-particle states Entangled-I such that, by exploiting spagtial orthogonality, they effectively correspond to separable (non-entangled) states of distinguishable particles) appears rather weak. As we have shown, exchangeability imposes constraints only on the columns, not on the rows. Hence, there is no reason why ``Entanglement I'' (finite exchangeability) should constrain the rows.

However, we agree with \cite{benatti2020entanglement}
 that, the equivalence in Theorem \ref{th:2} solely provides a mathematical representation of indistinguishability and, as such, we should not use it directly to define a notion of entanglement. Entanglement should not depend on the definition  of objective reality that the modeller has in their mind.
However, we believe that this mathematical representation (and  the analogy between classical finite exchangeability and quantum theory) can help to better understand and compare the various definitions of entanglement for indistinguishable particles. }

\section{Conclusions}
We have established a link between the mathematical representation of finite and infinite exchangeability in classical probability and the mathematical model of indistinguishable bosons in QT. We have demonstrated the parallelism between approaches from two `independent' research communities: probability theorists and quantum physicists. Both communities have developed related mathematical models, with similar structure and purpose, to tackle analogous problems, such as negative probabilities, symmetries and bounds with respect to the classical probability limit. They both have also obtained finite dimensional representations, via $L,\widehat{L}$ and $\rho$, for computing probabilities instead of relying on infinite dimensional objects such as the signed measure $\nu$.

We believe that what above points to a deeper and fundamental connection between the two areas. For instance, the representations in Propositions \ref{prop:kerns} and \ref{prop:thBern} can actually be seen as specific instances of representations in Generalised Probabilistic Theory \cite{birkhoff1936logic,mackey2013mathematical,jauch1963can,hardy2011foliable,hardy2001quantum,barrett2007information,chiribella2010probabilistic,barnum2011information,van2005implausible,pawlowski2009information,dakic2009quantum,fuchs2002quantum,brassard2005information,mueller2016information,coecke2012picturing,Caves02,Appleby05a,Appleby05b,Timpson08,longPaper,Fuchs&SchackII,mermin2014physics,pitowsky2003betting,Pitowsky2006,benavoli2016quantum,benavoli2017gleason,popescu1998causality,navascues2010glance,plavala2021general}. Moreover, as discussed in Section \ref{sec:application} (Application), the strong affinities between the two perspectives  can help to clarify the definition of entanglement for indistinguishable particles. As future work, in addition to digging into the aforementioned two research avenues, we  plan to explore how the interpretation of density matrices in QT as quasi-expectations might yield new insights into quantum information theory. For example, we believe that this interpretation and connection could prove useful in comprehending measures of uncertainty developed in quantum information theory, such as Von Neumann's entropy and quantum discord, and why, surprisingly in certain instances (quantum discord for separable density matrices), they deviate from their classical counterparts.
Finally, we intend to investigate the feasibility of extending our results to fermions.

\appendix

 \section{Quasi-expectation operator}
 \label{app:qexp}
 Consider a vector  of variables ${\bf x}$ taking value in the possibility space $\Omega$, and a vector space $\mathcal{F}$ of real-valued bounded functions  on ${\bf x}$ including the constants. Let $\mathcal{F}^+$ be the closed convex cone \footnote{A subset $\mathcal{C}$ of a real-vector space $\mathcal{F}$ is a cone if for each $f \in \mathcal{F}$ and positive scalar $\alpha>0$, the element  $\alpha f$ is in $\mathcal{C}$. A cone $\mathcal{C}$ is a convex cone if $\alpha f +\beta g$  belongs to $\mathcal{C}$, for any scalars  $\alpha,\beta>0$ and $f,g \in \mathcal{F}$.} of nonnegative functions in $\mathcal{F}$. 

An expectation operator is defined as follows.

\begin{definition}[{\cite[Sec.\ 2.8.4]{walley1991}}]\label{def:expect}
 Let $\widehat{L}$ be a  linear functional $\widehat{L}: \mathcal{F}\rightarrow \mathbb{R}$. $\widehat{L}$ is an \textbf{expectation operator}  if it satisfies the following property: 
\begin{equation}
\label{eq:Asup}
     \widehat{L}(g)\geq \sup c \text{ s.t. } g-c \in \mathcal{F}^+,  \tag{A} 
\end{equation}
for every $g \in \mathcal{F}$ and $c$ is the constant function of value $c$. 
\end{definition}
It can be easily verified that \eqref{eq:Asup} is equivalent to:
\begin{equation}
\label{eq:Ainf}
     L(g)\geq \inf_{{\bf x} \in \Omega} g({\bf x}).
\end{equation}
In the following, to simplify the notation, we simply write $\inf g$ instead of $\inf_{{\bf x} \in \Omega} g({\bf x})$. Linearity and \eqref{eq:Ainf} are the two properties that define a classical expectation operator. Indeed, from these two properties, we can derive that
\begin{itemize}
 \item $L(0)=0$;
 \item $0\stackrel{A}{=}L(0)=L(g-g)\stackrel{linearity}{=}L(g)+L(-g)$ and so $L(g)=-L(-g)$;
\end{itemize}
which, together with $L(-g)\geq \inf -g=-\sup g$, leads to
\begin{equation}
\label{eq:infsupbound}
 \inf g \leq L(g) \leq \sup g.
\end{equation}
This means that $L(g)$ is a `weighted-average' of $g$: the weights being the probability measure associated to the expectation operator; note in fact that $\inf g \leq \int_{\Omega} g\, dp \leq \sup g$ for any probability measure $p$. This formulation in terms of probabilities is not necessary. Indeed, we can more generally work with expectation operators.

A quasi-expectation operator is instead defined as follows.

\begin{definition}\label{def:TET}
 Let $\widehat{L}$ be a linear functional $\widehat{L}: \mathcal{F}\rightarrow \mathbb{R}$ and $\mathcal{C}^+$ be a closed convex cone (including the constants) such that $\mathcal{C}^+\subset \mathcal{F}^+$. We call $\hat{L}$ a \textbf{quasi-expectation} operator (QEO)  if it satisfies
 \begin{equation}
 \label{eq:Astar}
   \widehat{L}(g)\geq \sup c \text{ s.t. } g-c \in \mathcal{C}^+, \tag{A$^*$}   
 \end{equation}
 for every $g \in \mathcal{F}$.
\end{definition}
Let $\underline{c}_g$ be equal to the supremum value of $c$ such that $g-c \in \mathcal{C}^+$. It can be verified that \eqref{eq:Astar} implies:
 \begin{equation}
 \label{eq:inequcg}
   \widetilde{L}(g)\geq \underline{c}_g, \text{  where}   \inf_{{\bf x} \in \Omega} g({\bf x})\geq \underline{c}_g.   
 \end{equation}
 
A QEO (conservatively) relaxes  property \eqref{eq:Asup} by providing a lower bound  $\underline{c}_g$ to $\inf  g$.

From property \eqref{eq:Astar}, similarly to what was done for expectation operators, we can derive
\begin{equation}
\label{eq:infsupboundgen}
\underline{c}_g \leq \widetilde{L}(g) \leq \overline{c}_g,
\end{equation}
with $\overline{c}_g$ be equal to the infimum value of $c$ such that $c-g \in \mathcal{C}^+$, where 
\begin{equation}
\label{eq:infsupboundgen1}
\underline{c}_g  \leq \inf g  \leq \sup g \leq \overline{c}_g.
\end{equation}
Notice that since the external inequalities of Eq.~\eqref{eq:infsupboundgen1} can be strict for some $g$, we cannot in general define a QEO as an integral with respect to a probability measure and, therefore,  a QEO cannot be a `weighted average'. In other words, a QEO is not a classical expectation operator. In general, in order to  write a QEO as  an integral and satisfy \eqref{eq:infsupboundgen1}, we need to introduce some negative values:
$$
\widehat{L}(g)=\int_{\Omega} g\, d\nu,
$$
where $\nu$ is a signed-measure. 
 
\section{Proofs}

\subsection{Proof of Theorem \ref{th:repr41}}
We will first prove the first statement.

About the trace, by exploiting linearity of $\widehat{L}$, note that 
$$
Tr(M)=\widehat{L}\Big( (\otimes_r {\bf x})^\dagger I (\otimes_r {\bf x})\Big)
$$
where $I$ is the identity matrix. Remember that, by an abuse of notation, by $c$ we denote the constant polynomial of value $c$.
For $g({\bf x},{\bf x}^\dagger)=(\otimes_r {\bf x})^\dagger I (\otimes_r {\bf x})$, we have that $g-c$ is equal to $(1-c)(\otimes_r {\bf x})^\dagger I (\otimes_r {\bf x})$. This polynomial belongs to $\mathcal{S}_r^+$ provided that the matrix $(1-c)I$ is PSD. Therefore, $\sup c$, such that $(1-c)I$ is PSD, is equal to $c=1$. Similarly, we can consider the polynomial $-g-c$, with $g$ as before, and note that $\sup c$ is equal to $-1$. Therefore, we have shown that
$$
\widehat{L}(g)\geq 1, ~~\widehat{L}(-g)=-\widehat{L}(g)\geq -1
$$
which implies that $\widehat{L}(g)=1$. 
For the proof about $M$ being PSD, note that for any polynomial $(\otimes_r {\bf x})^\dagger G (\otimes_r {\bf x})$ with $G$ PSD, we have that 
$\widehat{L}(g)\geq0$. This implies that
$Tr(GM)\geq0 $ for all PSD matrices $G$. From a well-known result in linear algebra this implies that $M$ is PSD.

For the second part of the theorem, we first notice that 
$(\otimes_r {\bf x})^\dagger D (\otimes_r {\bf x})=\mathbf{1}^\top D(\otimes_r\boldsymbol{\theta})$, where  $\mathbf{1}$ is a column vector of ones of suitable dimension. From \eqref{eq:reprThsigned}, this allows us to rewrite
$$
P(t_1,\dots,t_r) = L(\theta_1^{n_1}\theta_2^{n_2}\cdots \theta_5^{n_5}\theta_6^{n_6})=L(\mathbf{1}^\top D(\otimes_r\boldsymbol{\theta})).
$$
We will now prove that, for each fixed $L$, there exists one $\widehat{L}$ such that $L(\mathbf{1}^\top D(\otimes_r\boldsymbol{\theta}))=\widehat{L}((\otimes_r {\bf x})^\dagger D (\otimes_r {\bf x}))$ for each diagonal $D$. This shows that we can represent any finitely exchangeable $P$ via $\widehat{L}$ and that, moreover, the relative $P$ is nonnegative and normalised.

The result is obtained by showing that, for any polynomial $g(\otimes_r\boldsymbol{\theta})=\mathbf{1}^\top D(\otimes_r\boldsymbol{\theta})=g({\bf x},{\bf x}^\dagger)=(\otimes_r {\bf x})^\dagger D (\otimes_r {\bf x})$
with $D$ diagonal, having defined
\begin{equation}
\label{eq:Aproof}
     L(g)\geq \sup c_1 \text{ s.t. } g-c_1 \in \mathcal{B}_r^+,  
\end{equation}
and
\begin{equation}
\label{eq:A1proof}
     \widehat{L}(g)\geq \sup c_2 \text{ s.t. } g-c_2 \in \mathcal{S}_r^+,  
\end{equation}
we have that $c_1=c_2$. 

This is immediate, because $g-c_2=(\otimes_r {\bf x})^\dagger(D-c_2I)(\otimes_r {\bf x})$. This polynomial is in the cone of $\mathcal{S}_r^+$ when the coefficients $diag(D-c_2I)$ are all nonnegative. The same condition, being nonnegative,  holds for the coefficients of the polynomial in $\boldsymbol{\theta}$. Note that, we can write any polynomial $g(\otimes_r\boldsymbol{\theta})$ as $\mathbf{1}^\top D(\otimes_r\boldsymbol{\theta})$, where  $\mathbf{1}$ is the column vector of ones. 
Therefore, $g-c_1=\mathbf{1}^\top (D-c_1I)(\otimes_r\boldsymbol{\theta})$. Since $\mathbf{1}^\top (D-c_1I)=diag(D-c_1I)$, the polynomial $g-c_1$ belongs to the cone $\mathcal{B}_r^+$ whenever the coefficients $diag(D-c_1I)$ are all nonnegative, equivalently when $D-c_1 I$ is PSD.
Then the statement follows by convexity of the cones $\mathcal{B}_r^+,\mathcal{S}_r^+$ and Proposition \ref{prop:thBern}.

\subsection{Equivalent definitions of Bose-symmetric density matrix}
\label{app:bosesymm}

\begin{proposition}
\label{eq:propoPiPi}
The following definitions of Bose-symmetric density matrix $\rho$ (for $r$ particles)  are equivalent:
\begin{align}
\label{eq:0a}
\rho&=\Pi_{sym} \rho \Pi_{sym}\\
\label{eq:0b}
\rho&=\Pi_i\rho=\rho\Pi_i^\dagger \text{ for each permutation matrix $\Pi_i$}.
\end{align}
\end{proposition}
We start  the proof by showing that condition \eqref{eq:0b} implies condition \eqref{eq:0a}. From $\rho=\Pi_i\rho$ for each $\Pi_i$, we obtain that $
\rho= \frac{1}{N!} \sum_{i} \Pi_i\rho$. Then, by  the definition \eqref{eq:symm} stating that $\Pi_{sym}:=\frac{1}{N!} \sum_{i} \Pi_i$
it follows that
\begin{align}
\label{eq:1p}
\rho=\Pi_{sym} \rho.
\end{align}
Similarly, from $\rho=\rho\Pi_i^\dagger$ and the fact that $(\cdot)^\dagger$ distributes over sum, for each $\Pi_i$, 
we get that $
\rho= \rho(\frac{1}{N!} \sum_{i}\Pi_i^\dagger )= \rho(\frac{1}{N!} \sum_{i}\Pi_i )^\dagger$. By applying  
definition \eqref{eq:symm} and since $\Pi_{sym}^\dagger= \Pi_{sym}$ (see for instance  \cite[B-2-c]{cohen2020quantum}), from \eqref{eq:0b}, we have that
\begin{align}
\label{eq:2p}
\rho=\rho\Pi_{sym}.
\end{align}

From equation\eqref{eq:1p}, by multiplying both sides by $\Pi_{sym}$ from the right, we obtain that
$\rho \Pi_{sym}=\Pi_{sym}\rho \Pi_{sym}$. Finally, using Eq.~\eqref{eq:2p}, we obtain that:
 $$
 \rho = \Pi_{sym}\rho \Pi_{sym}.
 $$
 
To verify that  condition \eqref{eq:0a} implies condition \eqref{eq:0b}, let us assume the former. But since $\Pi_i \Pi_{sym}=\Pi_{sym}$ (see for instance  \cite[B-2-c]{cohen2020quantum}), $\Pi_{sym}^\dagger= \Pi_{sym}$ and therefore $\Pi_{sym}=(\Pi_i \Pi_{sym})^\dagger=\Pi_{sym}^\dagger\Pi_i^\dagger = \Pi_{sym}\Pi_i^\dagger$, it is immediate to
 derive that
 $$
 \Pi_i \rho =  \Pi_i \Pi_{sym}\rho \Pi_{sym}=\Pi_{sym}\rho \Pi_{sym}=\rho, 
 $$
 and that
 $$
 \rho\Pi_i^\dagger  =  \Pi_{sym}\rho \Pi_{sym}\Pi_i^\dagger =\Pi_{sym}\rho \Pi_{sym}=\rho. 
 $$
     
\subsection{Proof Theorem \ref{th:2}}\label{app:thm2}
In order to prove Theorem \ref{th:2}, we need to introduce some notation and state a lemma.

First, we recall some results from \cite{Benavoli2021f}. They will allow  us to rewrite
any density matrix of $r$-particles as
$\rho=\widehat{L}((\z_1\otimes\z_2\otimes\cdots\otimes\z_r)(\z_1\otimes\z_2\otimes\cdots\otimes\z_r)^\dagger)$
where  $\z_i \in \mathbb{C}^n$ (n=2 for a qubit) such that $\z_i^\dagger \z_i=1$, where $\widehat{L}$ is a linear operator from the monomials represented by the elements of the matrix $(\z_1\otimes\z_2\otimes\cdots\otimes\z_r)(\z_1\otimes\z_2\otimes\cdots\otimes\z_r)^\dagger$ to $\mathbb{C}$, such that the matrix $\widehat{L}((\z_1\otimes\z_2\otimes\cdots\otimes\z_r)(\z_1\otimes\z_2\otimes\cdots\otimes\z_r)^\dagger)$ is PSD and with trace one ($\widehat{L}$ preserves the constants).

We can equivalently rewrite the elements of $\rho_{ij}$ as
$ (y_{\boldsymbol{\alpha},\boldsymbol{\beta}})_{\boldsymbol{\alpha},\boldsymbol{\beta}}:=(\widehat{L}(\z^{\boldsymbol{\alpha}}(\z^\dagger)^{\boldsymbol{\beta}}))_{\boldsymbol{\alpha},\boldsymbol{\beta}}$
where $\z:=[\z_1^\top, \z_2^\top,\dots,\z_r^\top]^\top$ and $\boldsymbol{\alpha}$ and $\boldsymbol{\beta}$ are a binary encoding of the indexes $i$ and, respectively $j$.  In particular, we have $\boldsymbol{\alpha}:=[\boldsymbol{\alpha}_1,\dots,\boldsymbol{\alpha}_{r}]$ where  $\boldsymbol{\alpha}_k$ is a vector in $\mathbb{N}^n$ which has all zero entries expect one element which is equal to one. The notation  $\z^{\boldsymbol{\alpha}}$ must be interpreted as $\prod_{i=1}^r\prod_{j=1}^n z_{ij}^{\alpha_{ij}}$.
\begin{example}
    For instance, for $n=2$  and $r=2$ (two particles), we have that
$\z:=[z_{11},z_{12},z_{21},z_{22}]^\top$. We can then rewrite for instance $\rho_{11}$ as 
$\widehat{L}(z_{11}z_{21}(z_{11}z_{21})^\dagger)=\widehat{L}(\z^{\boldsymbol{\alpha}}(\z^\dagger)^{\boldsymbol{\beta}})=\widehat{L}(\z^{[1,0,1,0]}(\z^\dagger)^{[1,0,1,0]})$, where
$\boldsymbol{\alpha}=[\boldsymbol{\alpha}_1,\boldsymbol{\alpha}_2]$
with $\boldsymbol{\alpha}_1=[1,0]$ and $\boldsymbol{\alpha}_2=[1,0]$, and $\boldsymbol{\beta}=\boldsymbol{\alpha}$.
Note that, we have
$$
\begin{aligned}
  &\z^{[1,0,1,0]}=z_{11}^1 z_{12}^0z_{21}^1z_{22}^0=z_{11}z_{21}\\
   &(\z^\dagger)^{[1,0,1,0]}=z^\dagger_{11}z^\dagger_{21}.\\ 
\end{aligned}
$$
Introducing the notation $ (y_{\boldsymbol{\alpha},\boldsymbol{\beta}})_{\boldsymbol{\alpha},\boldsymbol{\beta}}:=(\widehat{L}(\z^{\boldsymbol{\alpha}}(\z^\dagger)^{\boldsymbol{\beta}}))_{\boldsymbol{\alpha},\boldsymbol{\beta}}$, we have that the elements of the matrix $\rho$, for the example  $n=r=2$, can be denoted as follows:
\begin{align}
\label{eq:rhoy}
\nonumber
\rho=&\left[\begin{smallmatrix}
y_{[1,0,1,0],[1,0,1,0]} & y_{[1,0,1,0],[1,0,0,1]} & y_{[1,0,1,0],[0,1,1,0]}  & y_{[1,0,1,0],[0,1,0,1]} \\
y_{[1,0,0,1],[1,0,1,0]} & y_{[1,0,0,1],[1,0,0,1]} & y_{[1,0,0,1],[0,1,1,0]}  & y_{[1,0,0,1],[0,1,0,1]} \\
y_{[0,1,1,0],[1,0,1,0]} & y_{[0,1,1,0],[1,0,0,1]} & y_{[0,1,1,0],[0,1,1,0]}  & y_{[0,1,1,0],[0,1,0,1]} \\
y_{[0,1,0,1],[1,0,1,0]} & y_{[0,1,0,1],[1,0,0,1]} & y_{[0,1,0,1],[0,1,1,0]}  & y_{[0,1,0,1],[0,1,0,1]} \\
\end{smallmatrix}\right].
\end{align}
where $y_{[1,0,1,0],[1,0,1,0]}=\widehat{L}(z_{11}z_{21}(z_{11}z_{21})^\dagger)$ and so on.
\end{example}

If $\rho$ is the density matrix of $r$ indistinguishable bosons, then it has to satisfy the constraint $\Pi_{sym} \rho \Pi_{sym}= \rho$. We can then prove the following lemma.
\begin{lemma}
\label{lem:appy}
Let $\rho_{\boldsymbol{\alpha},\boldsymbol{\beta}}=\widehat{L}(\z^{\boldsymbol{\alpha}}(\z^\dagger)^{\boldsymbol{\beta}})=y_{\boldsymbol{\alpha},\boldsymbol{\beta}}$. If $\rho$ satisfies $\Pi_{sym} \rho \Pi_{sym}= \rho$ then for all the indexes $\boldsymbol{\alpha},\boldsymbol{\beta}$, the elements $\rho_{\boldsymbol{\alpha},\boldsymbol{\beta}}$ have to satisfy the following constraints:
$$
\begin{aligned}
&y_{[\boldsymbol{\alpha}_1,\dots,\boldsymbol{\alpha}_r],[\boldsymbol{\beta}_1,\dots,\boldsymbol{\beta}_r]}=y_{[\boldsymbol{\alpha}_1,\dots,\boldsymbol{\alpha}_r],[\boldsymbol{\beta}_{\pi(1)},\dots,\boldsymbol{\beta}_{\pi(r)}]}\\&=y_{[\boldsymbol{\alpha}_{\pi(1)},\dots,\boldsymbol{\alpha}_{\pi(r)}],[\boldsymbol{\beta}_1,\dots,\boldsymbol{\beta}_r]}   =y_{[\boldsymbol{\alpha}_{\pi(1)},\dots,\boldsymbol{\alpha}_{\pi(r)}],[\boldsymbol{\beta}_{\pi(1)},\dots,\boldsymbol{\beta}_{\pi(r)}]}  
\end{aligned}
$$
for any permutation $\pi$ of the labels of the variables.
\end{lemma}
\begin{proof}
 This follow from the following known property of the symmetriser:
$$
\Pi_{sym} P_{\pi} = P_{\pi} \Pi_{sym}= \Pi_{sym}
$$
where $P_{\pi}$ is any  permutation matrix (defined by  the permutation $\pi$ of the labels) associated with the $r$ particles.  { Note that, previously in the manuscript, we denoted permutation matrices with  $\Pi_i$ instead of $P_{\pi}$. Here we use $P_{\pi}$ to state that this is the permutation matrix relative to the label permutation $\pi$. } Given that $\rho=\Pi_{sym} \rho \Pi_{sym}$, this means that
$$
\begin{aligned}
P_{\pi}\rho&=P_{\pi}\Pi_{sym}\rho\Pi_{sym}=\Pi_{sym}\rho\Pi_{sym} P_{\pi}\\
&=P_{\pi}\Pi_{sym}\rho\Pi_{sym} P_{\pi}=\Pi_{sym}\rho\Pi_{sym}=\rho.
\end{aligned}
$$
Therefore, by linearity of $\widehat{L}$, we have that
$$
\begin{aligned}
&\rho:= \widehat{L}( (\z_1\otimes\z_2\otimes\cdots\otimes\z_r)(\z_1\otimes\z_2\otimes\cdots\otimes\z_r)^\dagger)\\&= P_{\pi}\widehat{L}( (\z_1\otimes\z_2\otimes\cdots\otimes\z_r)(\z_1\otimes\z_2\otimes\cdots\otimes\z_r)^\dagger)\\
&= \widehat{L}( P_{\pi}(\z_1\otimes\z_2\otimes\cdots\otimes\z_r)(\z_1\otimes\z_2\otimes\cdots\otimes\z_r)^\dagger)\\
&= \widehat{L}((\z_1\otimes\z_2\otimes\cdots\otimes\z_r)(\z_1\otimes\z_2\otimes\cdots\otimes\z_r)^\dagger P_{\pi})\\&= \widehat{L}(P_{\pi}(\z_1\otimes\z_2\otimes\cdots\otimes\z_r)(\z_1\otimes\z_2\otimes\cdots\otimes\z_r)^\dagger P_{\pi})    
\end{aligned}
$$
for every $P_{\pi}$. Note that, the  permutation matrix rearranges the $\boldsymbol{\alpha}$-th element of the vector $\z_1\otimes\z_2\otimes\cdots\otimes\z_r$ according to the following rule:
$$
\begin{aligned}
   &P_{\pi}( \z_1\otimes\cdots\otimes\z_r)_{[\boldsymbol{\alpha}_1,\dots,\boldsymbol{\alpha}_r]} =( \z_1\otimes\cdots\otimes\z_r)_{[\boldsymbol{\alpha}_{\pi(1)},\dots,\boldsymbol{\alpha}_{\pi(r)}]}\\
&=\z^{[\boldsymbol{\alpha}_{\pi(1)},\dots,\boldsymbol{\alpha}_{\pi(r)}]} 
\end{aligned}
$$
which ends the proof.
\end{proof}
For the example in \eqref{eq:rhoy}, the results of the above lemma implies 
for instance that $$\begin{aligned}&y_{[1,0,0,1],[1,0,0,1]}=y_{[1,0,0,1],[0,1,1,0]}\\
&=y_{[0,1,1,0],[1,0,0,1]}=y_{[0,1,1,0],[0,1,1,0]}.\end{aligned}$$

We are now ready to prove Theorem \ref{th:2}. First we prove that $\mathcal{S}_2 \subseteq\mathcal{S}_1$. Note that
\begin{align}
\nonumber
    \Pi_{sym} M \Pi_{sym}&=\widehat{L}\Big( \Pi_{sym}(\otimes_r {\bf x})(\otimes_r {\bf x})^\dagger\Pi_{sym}^\dagger\Big)\\
    &=\widehat{L}\Big( (\otimes_r {\bf x})(\otimes_r {\bf x})^\dagger\Big)=M,
\end{align}
where we have exploited \eqref{eq:symm} to derive  $\Pi_{sym}(\otimes_r {\bf x}) = \otimes_r {\bf x}$. Note in fact that, $\otimes_r {\bf x}$ is clearly symmetric to permutations, being the r-times tensor product of ${\bf x}$. 

We will now prove that $\mathcal{S}_1 \subseteq\mathcal{S}_2$. This part is based on the same insight that motivates the second quantisation formalism in QT. We start from Lemma \ref{lem:appy} and note that in
$$
\begin{aligned}
&y_{[\boldsymbol{\alpha}_1,\dots,\boldsymbol{\alpha}_r],[\boldsymbol{\beta}_1,\dots,\boldsymbol{\beta}_r]}=y_{[\boldsymbol{\alpha}_1,\dots,\boldsymbol{\alpha}_r],[\boldsymbol{\beta}_{\pi(1)},\dots,\boldsymbol{\beta}_{\pi(r)}]}\\&=y_{[\boldsymbol{\alpha}_{\pi(1)},\dots,\boldsymbol{\alpha}_{\pi(r)}],[\boldsymbol{\beta}_1,\dots,\boldsymbol{\beta}_r]}   =y_{[\boldsymbol{\alpha}_{\pi(1)},\dots,\boldsymbol{\alpha}_{\pi(r)}],[\boldsymbol{\beta}_{\pi(1)},\dots,\boldsymbol{\beta}_{\pi(r)}]}  
\end{aligned}
$$
the permutation acts on the indexing vectors. Each  $\boldsymbol{\alpha}_i$ is an indicator vector which denotes a component of the vector $\z_i$ and the permutation simply swaps this indicator vectors. We now introduce the occupation number vector $[k_1,k_2,\dots,k_n]$ which, for a given $\boldsymbol{\alpha}_1,\dots,\boldsymbol{\alpha}_r$, counts the number of variables $\z_i$ that are in their state 1, state 2,\dots, or state $n$. All the vectors $[\boldsymbol{\alpha}_{\pi(1)},\dots,\boldsymbol{\alpha}_{\pi(r)}]$ have the same occupation number. For instance, for the previous example $y_{[1,0,0,1],[1,0,0,1]}=y_{[1,0,0,1],[0,1,1,0]}=y_{[0,1,1,0],[1,0,0,1]}=y_{[0,1,1,0],[0,1,1,0]}$. The occupation number for $[1,0,0,1]$ is $[1,1]$ because variable 1 is in state 1 and variable 2 in state 2. $[0,1,1,0]$ has the same occupation number.
Therefore, we can write 
 $y_{[1,0,0,1],[1,0,0,1]}=y_{[1,0,0,1],[0,1,1,0]}=y_{[0,1,1,0],[1,0,0,1]}=y_{[0,1,1,0],[0,1,1,0]} \rightarrow y_{[k_1,k_2][k'_1,k'_2]}$
 with $k_1,k_2;k'_1,k'_2=1$,  where we used $'$ to denote the occupation number relative to $\boldsymbol{\beta}$.
 The representation of   $y_{[1,0,1,0],[1,0,1,0]}$ using the occupation number as index is $y_{[2,0],[2,0]}$.

 We conclude the proof by noticing that we can represent the elements of $\widehat{L}( (\otimes_r {\bf x})(\otimes_r {\bf x})^\dagger)$ as 
 $(\widehat{L}( {\bf x}^{\bf k}({\bf x}^\dagger)^{{\bf k}'}))_{{\bf k},{\bf k}'}=y_{{\bf k},{\bf k}'}$, where
 ${\bf k}=[k_1,\dots,k_n]$ is such that each $k_i \in \mathbb{N}$ denotes the degree of the i-th variable in the monomial. For instance, $\widehat{L}(x_1x_2(x_1x_2)^\dagger)=y_{[1,1],[1,1]}$ and $\widehat{L}(x_1^2(x_1^2)^\dagger)=y_{[2,0],[2,0]}$.

\bibliographystyle{elsarticle-num-names} 
 \bibliography{biblio}%

\end{document}